\documentclass[a4paper,11pt]{article}

\title{Indirect Inference for L\'evy-driven continuous-time GARCH models}
  
\author{Thiago do R\^ego Sousa\thanks{Center for Mathematical Sciences, Technical University of Munich,  85748 Garching, Boltzmannstr.~3, Germany, e-mail:  thiago.sousa@tum.de, cklu@tum.de, haug@tum.de}
\and Stephan Haug\footnotemark[1]
\and Claudia Kl\"uppelberg\footnotemark[1]
}

       \usepackage{graphicx}
       \usepackage{braket} 
       \usepackage{amsmath} 
       \usepackage{bm} 
       \usepackage{amsfonts,amsthm,amstext,amssymb,amsmath}
       \usepackage[toc,page]{appendix}
       \usepackage{pdfpages}
       \usepackage{caption}
       \usepackage{color}
       \usepackage{comment}
       \usepackage[numbers]{natbib}
       \usepackage{blkarray} 
	   \usepackage{multirow} 
	   \usepackage{dsfont}
	   \usepackage{enumitem}
	   \usepackage{kantlipsum} 

\numberwithin{equation}{section}

\newtheorem{theorem}{Theorem}[section]
\newtheorem{lemma}[theorem]{Lemma}
\newtheorem{remark}[theorem]{Remark}
\newtheorem{example}[theorem]{Example}
\newtheorem{proposition}[theorem]{Proposition}
\newtheorem{definition}[theorem]{Definition}

\newtheorem{corollary}[theorem]{Corollary}

\newcommand{\bthe}{\begin{theorem}}
\newcommand{\ethe}{\end{theorem}}

\newcommand{\ble}{\begin{lemma}}
\newcommand{\ele}{\end{lemma}}

\newcommand{\bde}{\begin{definition}\rm}
\newcommand{\ede}{\Chalmos\end{definition}}

\newcommand{\bco}{\begin{corollary}}
\newcommand{\eco}{\end{corollary}}

\newcommand{\bpr}{\begin{proposition}}
\newcommand{\epr}{\end{proposition}}

\newcommand{\brem}{\begin{remark}\rm}
\newcommand{\erem}{\Chalmos\end{remark}}

\newcommand{\bproof}{\begin{proof}}
\newcommand{\eproof}{\end{proof}}

\newcommand{\bexam}{\begin{example}\rm}
\newcommand{\eexam}{\Chalmos\end{example}}

\newcommand{\beao}{\begin{eqnarray*}}
\newcommand{\eeao}{\end{eqnarray*}\noindent}

\newcommand{\beam}{\begin{eqnarray}}
\newcommand{\eeam}{\end{eqnarray}\noindent}

\newcommand{\barr}{\begin{array}}
\newcommand{\earr}{\end{array}}

\def\N{{\mathbb N}}

\def\Z{{\mathbb Z}}
\def\P{{\mathbb P}}
\def\E{{\mathbb E}}
\def\R{{\mathbb R}}

\def\Q{{\mathbb Q}}

\def\calf{{\mathcal{F}}}
\def\calh{{\mathcal{H}}}
\def\call{{\mathcal{L}}}

\def\calm{{\mathcal{M}}}

 \def\1{\mathds{1}}

\newcommand{\stp}{\stackrel{P}{\rightarrow}}
\newcommand{\std}{\stackrel{d}{\rightarrow}}
\newcommand{\stas}{\stackrel{\rm a.s.}{\rightarrow}}

\newcommand{\eqd}{\stackrel{\mathrm{d}}{=}}

\newcommand{\tto}{{t\to\infty}}

\newcommand{\nto}{{n\to\infty}}

\newcommand{\htzero}{{h\to 0}}

\newcommand{\si}{{\sigma}}

\newcommand{\vp}{\varphi}

\DeclareMathOperator*{\argminA}{arg\,min} 
\newcommand{\var}{\mathbb{V}{\rm ar}}
\newcommand{\cov}{\mathbb{C}{\rm ov}}

\newcommand{\ls}{{\rm LS}}

\newcommand{\yw}{{\rm YW}}
\newcommand{\ii}{{\rm II}}
\newcommand{\iistar}{{\rm II}^\ast}
\newcommand{\mm}{{\rm MM}}
\newcommand{\opb}{{\rm OPB}}

\newcommand{\bstheta}{{\bm{\theta}}}
\newcommand{\bsTheta}{{\Theta}}

\newcommand{\bsa}{{\bm{a}}}
\newcommand{\bspi}{{\bm{\pi}}}

\newcommand{\bsx}{{\bm{x}}}
\newcommand{\hatbspi}{\hat{\bm{\pi}}}
\newcommand{\tildebsw}{{\tilde{\bm{W}}_n}}
\newcommand{\tildebsy}{{\tilde{\bm{Y}}_n}}
\newcommand{\tildevp}{{\tilde{\vp}}}
\newcommand{\barbsw}{{\bar{\bm{W}}_n}}
\newcommand{\barbsy}{{\bar{\bm{Y}}_n}}        
\newcommand{\sumzerous}{\sum_{0 < u \leq s}}
\newcommand{\sumzerout}{\sum_{0 < u \leq t}}

\newcommand{\diff}{{\rm d}}

\newcommand{\gcont}{G^{(c)}}
\newcommand{\xcont}{X^{(c)}}

\newcommand{\ThetaHCone}{\Theta^{(p(1+\epsilon))}}
\newcommand{\ThetaHCtwo}{\Theta^{(2p(1+\epsilon))}}

\newcommand{\supThetaOnePLoc}{\sup_{\substack{0 < \|\bstheta_1-\bstheta_2\| < 1 \\ \bstheta_1,\bstheta_2 \in \ThetaHCone}}}
\newcommand{\supThetaTwoPLoc}{\sup_{\substack{0 < \|\bstheta_1-\bstheta_2\| < 1 \\ \bstheta_1,\bstheta_2 \in \ThetaHCtwo}}}

\newcommand{\gradtheta}{\nabla_{\bstheta}}
\newcommand{\gradtwo}{\nabla_{\eta,\varphi}}

\newcommand{\partiall}{\frac{\partial}{\partial \bstheta_l}}

\newcommand{\partialkl}{\frac{\partial^2}{\partial \bstheta_k \bstheta_l}}

\newcommand{\ov}{\overline}

\newcommand{\Chalmos}{\quad\hfill\mbox{$\Box$}}  

\allowdisplaybreaks[3]

\begin{document}
\maketitle

\begin{abstract}
We advocate the use of an Indirect Inference method to estimate the parameter of a COGARCH(1,1) process for equally spaced observations. 
This requires that the true model can be simulated and a reasonable estimation method for an approximate auxiliary model. 
We follow previous approaches and use linear projections leading to an auxiliary autoregressive model for the squared COGARCH returns. 
The asymptotic theory of the Indirect Inference estimator relies {on a uniform SLLN and asymptotic normality of the parameter estimates of the auxiliary model,
which require continuity and differentiability of the COGARCH process} with respect to its parameter and which we prove via Kolmogorov's continuity criterion. 
This leads to consistent and asymptotically normal Indirect Inference estimates under moment conditions on the driving L\'evy process. 
A simulation study shows that the method yields a substantial finite sample bias reduction compared to previous estimators. 
\end{abstract}


\section{Introduction}

The COGARCH(1,1) process was introduced in \citet{Kluppelberg04} as a continuous time analog of the discrete time GARCH(1,1) process. 
It is defined as
\begin{equation}\label{eq:def:Pt}
P_t(\bstheta) = \int_0^t \sigma_s(\bstheta) \diff L_s, \quad t\ge 0,
\end{equation}
with parameter $\bm{\theta}$ (to be specified in Section~\ref{s2}), $L$ is a L\'evy process with L\'{e}vy measure $\nu_L\not\equiv 0$ and having c\`adl\`ag  sample paths.
The volatility process $(\sigma_s(\bstheta))_{s \geq 0}$ is predictable and its stochasticity depends only on $L$. 
The COGARCH process satisfies many stylized features of financial time series and is suited for modeling high-frequency data (see \citet{Bayraci14}, \citet{Bibbona15}, \citet{Haug07}, \citet{Maller08}, \citet{Kluppelberg10}, and \citet{Muller10}).

In many practical problems, one observes the log-price process $(P_{i\Delta}(\bstheta_0))_{i=1}^n$ on a
fixed grid of size $\Delta>0$ and the question of interest is how to estimate the true parameter $\bstheta_0$.
The data used for estimation are returns $(G_i(\bstheta_0))_{i=1}^n$, where
\begin{equation}\label{eq:def:Gi}
G_i(\bstheta_0) := P_{{\Delta i}}(\bstheta_0) - P_{{(i-1)\Delta}}(\bstheta_0) = \int_{(i-1)\Delta}^{i\Delta} \sigma_s(\bstheta_0) \diff L_s.
\end{equation}
Several methods have been proposed to estimate the parameter of a COGARCH process. 
A method of moments was proposed in \citet{Haug07}, 
\citet{Bibbona15} used prediction based estimation as developed in \citet{Sorensen00}, and 
\citet{Maller08} proposed a pseudo maximum likelihood (PML) method which also works for non-equally spaced observations. 
Both moment and prediction based estimators  are consistent and asymptotically normal under certain regularity conditions.
The asymptotic properties of the PML estimator were studied in \citet{Iannace14} and in \citet{Kim13}, which require that $\Delta\downarrow 0$ as $n\to\infty$. For the COGARCH process, \citet{Bayraci14} used Indirect Inference  with an auxiliary discrete-time GARCH model with Gaussian residuals. No theoretical results were proved, but their simulation study suggests that Indirect Inference estimators achieve a similar performance as the PML estimator of \citet{Maller08} for fixed $\Delta>0$. Furthermore, \citet{Muller10} proposed a Markov chain Monte Carlo method, when $L$ is a compound Poisson process. 

In this paper we advocate an Indirect Inference method, different to the one suggested in \citet{Bayraci14}, to estimate the COGARCH parameter and derive the asymptotic properties of the estimator.
Such methods were introduced in \citet{Smith93} and generalized in \citet{Gourieroux93}, and offer a way to overcome many estimation problems by a clever simulation method. 
In short, it only requires that the true model can be simulated and a reasonable estimation method for an approximate auxiliary model. 

Indirect Inference was originally introduced for complex econometric models to overcome the estimation problem of an intractable likelihood function, as e.g. for continuous time models with stochastic volatility (see \citet{Bianchi96}, \citet{Jiang98}, \citet{Laurini13}, \citet{Raknerud12}, and \citet{Wahlberg15}). 
Indirect Inference can also be used as a vehicle to produce estimators which are robust, when there are outliers in the observations (see \citet{deLuna01} for robust estimation of a discrete time ARMA and \citet{Fasen18II} of a continuous time ARMA). 
Another motivation is given in \citet{Gourieroux00, Gourieroux10}, where it is shown that Indirect Inference can reduce the finite sample bias considerably.
This is our motivation to study the asymptotic properties of Indirect Inference estimators (IIE) in the context of COGARCH estimation.

The Indirect Inference procedure works as follows. Let $\bspi$ denote the parameter of an auxiliary model chosen for the COGARCH returns $(G_i(\bstheta_0))_{i=1}^n$ or some transformed random variables. 
From this data we estimate $\bspi$ and obtain $\hatbspi_n$. 
For many different $\bstheta \in \Theta$ we simulate $K\ge 1$ independent samples of size $n$ of COGARCH returns $(G_i^{(k)}(\bstheta))_{i=1}^n$ and compute the estimators $\hatbspi_{n,k}(\bm{\theta})$ for $k = 1,\dots,K$. 
The IIE of $\bstheta$ is then defined as
\begin{equation}\label{eq:1.3}
\hat{\bm{\theta}}_{n,\ii} := \argminA_{\bm{\theta} \in \Theta}  \bigg(\hatbspi_n - \frac{1}{K} \sum_{k=1}^K \hatbspi_{n,k}(\bstheta) \bigg)^\top \bm{\Omega}\bigg(\hatbspi_n - \frac{1}{K} \sum_{k=1}^K \hatbspi_{n,k}(\bstheta) \bigg),
\end{equation}
where $\bm{\Omega}$ is a symmetric and positive definite weight matrix.
Under certain regularity conditions, IIEs are consistent and asymptotically normal. 
These regularity conditions are mainly related to three aspects: 
(A) find an auxiliary model whose parameter is connected to the COGARCH parameter through a one-to-one binding function, (B) prove strong consistency and asymptotic normality of $\hatbspi_n$, and (C) prove that the estimator $\hatbspi_{n}(\bstheta)$, as a function of $\bstheta$, satisfies conditions for the application of a uniform strong law of large numbers (SLLN) and a delta method for the asymptotic normality.

The starting point (A) is an appropriate auxiliary model that provides a one-to-one binding function. 
We follow previous approaches and use linear projections leading to an auxiliary autoregressive (AR) model of appropriate order for the squared COGARCH returns $(G^2_i(\bstheta))_{i \in \N}$. 
Often the properties of the binding function are assessed via simulation (see \citet{Lombardi08} and \citet{Garcia11}), but for our models the binding function can be proved to be one-to-one.

Part (B), strong consistency and asymptotic normality of the estimator $\hatbspi_n$ of the AR model parameter $\bspi$, is obtained in a similar way as in classical time series analysis (see e.g. \citet{Brockwell13}), extending the theory to residuals, which may not be white noise, but an arbitrary stationary and ergodic process with finite variance.
The SLLN and asymptotic normality of $\hatbspi_n$ will then be a consequence of the fact that  $(G^2_i(\bstheta))_{i \in \N}$ is also strong mixing with appropriate mixing coefficients.

(C) is related to regularity conditions of the map $\bstheta \mapsto \hatbspi_{n}(\bstheta)$. 
To achieve strong consistency of the IIE we need to show that
\begin{equation*}
\sup_{\bm{\theta} \in \Theta}  \| \hatbspi_{n}(\bstheta) - \bm{\pi}_{\bstheta}  \| \overset{\text{a.s.}}{\rightarrow} 0, \quad \nto.
\end{equation*}
To move from point-wise to uniform convergence we use a uniform SLLN in a compact parameter space $\Theta$. 
For the estimator we study here, the application of a uniform SLLN holds provided that $G_i(\bm{\theta})$ is a continuous function in $\bm{\theta}$ and $
\E \sup_{\bm{\theta} \in \Theta}  G^4_i(\bm{\theta}) < \infty$ for all $i \in \N$. 
The continuity of this map does not follow directly from the continuity of $\sigma_s(\bstheta)$ for fixed $s$, because the L\'{e}vy process in the stochastic integral in \eqref{eq:def:Gi} may have infinite variation. 
Under conditions on the moments and the characteristic exponent of the driving L\'{e}vy process, we find a version of $G_i(\bstheta)$ which is continuous by Kolmogorov's continuity criterion, and as a result we conclude strong consistency of the IIE $\hat{\bm{\theta}}_{n,\ii}$. 
A Taylor expansion of $\hatbspi_n(\bstheta)$ around the true parameter $\bstheta_0$ yields asymptotic normality by the delta method.
This will require continuous differentiability of $G_i(\bstheta)$ in $\bstheta$, which will follow from a result of \citet{James84} together with Kolmogorov's continuity criterion. 

Our paper is organised as follows. 
We start in Section 2 with the formal definition of a stationary COGARCH process as returns process, and recall its relevant properties.
We also present the autoregressive auxiliary model of the squared returns and define the least squares estimator (LSE) and Yule-Walker estimator (YWE) of the AR parameter, as well as  the binding function giving the link to the COGARCH parameter. In Section~2.3 we present the IIE and the conditions, which guarantee a uniform SLLN and asymptotic normality of the IIE.
In Section~3 we prove strong consistency and asymptotic normality of the LSE and YWE under the non-standard conditions of stationary ergodicity and a mixing property. 
Section~4 is dedicated to strong consistency and asymptotic normality of the IIE of the COGARCH process.
Section~5 presents a simulation study and shows that the bias reduction based on the IIE is indeed substantial compared to previous estimators.
Technical results like conditions for the existence of a version of the COGARCH returns, which is continuous in its parameter and other auxiliary results are summarized in an Appendix.

Throughout we write $\|\cdot\|$ for the $\ell^1$-norm in $\R^d$ for $d\in\N$ and recall that in $\R^d$ all norms are equivalent. For a matrix $A \in \R^{p\times q}$ we also write $\|A\|$ for the matrix norm generated by the $\ell^1$-norm.
For a vector $x \in \R^d$ and a $d\times d$  positive definite matrix $\bm{\Omega}$ we write $\| x \|_{\bm{\Omega}} = x^\top \Omega x $.
Furthermore,  we denote by $\call^p$ the space of $p$-integrable random variables, and  by $\dim(A)$ the dimension of a subset $A$ of $\R^d$. 
For a function $f(\bstheta)$ in $\R$ with $\bstheta \in \R^q$ the gradient with respect to $\bstheta$ is $\gradtheta f(\bstheta) = (\partiall f(\bstheta))_{l=1}^q \in \R^{q}$ and $\gradtheta^2 f(\bstheta) = (\partialkl f(\bstheta))_{k,l=1}^q \in \R^{q \times q}$ denotes the Hessian matrix.

\section{COGARCH process, auxiliary AR representation and Indirect Inference Estimation}\label{s2}

\subsection{Definition of the COGARCH process}

For the parameter space of the COGARCH process given as $\{\bstheta=(\beta,\eta,\vp)^\top : \beta,\eta,\vp>0\}$, we construct a strictly stationary version of the volatility process as in \citet{Kluppelberg04}.
First define the process $(Y_s(\bstheta))_{s \geq0}$ by
\begin{equation}\label{eq:def:yt}
Y_s(\bm{\theta}) := \eta s - \sumzerous \log ( 1 + \vp (\Delta L_u)^2 ), \quad s \geq 0,
\end{equation}
with Laplace transform  $\E e^{-pY_s(\bm{\theta})} = e^{ s \Psi_{\bm{\theta}}(p) }$, where
\begin{equation}\label{eq:mgf:yt}
\Psi_{\bm{\theta}}(p) = -p\eta + \int_{\mathbb{R}} ( (1 + \vp x^2)^p - 1 ) \nu_L(\diff x), \quad p \geq 0.
\end{equation} 
We shall often use the fact that for $p>0$ by Lemma~4.1(a) in \cite{Kluppelberg04},
$$\E |L_1|^{2p} < \infty\quad\mbox{if and only if}\quad |\Psi_\bstheta(p)| <\infty.$$
Define the volatility process $(\sigma^2_t(\bstheta))_{t \geq 0}$ by
\begin{equation}\label{eq:def:sigt}
\sigma^2_t(\bm{\theta}) :=  \Big( \beta\int_0^t e^{Y_s(\bm{\theta})} \diff s + \sigma^2_0(\bm{\theta}) \Big) e^{-Y_{t-}(\bm{\theta})}, \quad t \geq 0,
\end{equation}
where $Y_{t-}(\bstheta)$ denotes the left limit at $t$ and $\sigma^2_0(\bm{\theta})$ the  starting value of the volatility process. 
If $\E |L_1|^{2} < \infty$  and $\Psi_{\bm{\theta}}(1) < 0$, then by Lemma~4.1(c) of \cite{Kluppelberg04},
$\sigma^2_t(\bm{\theta}) \overset{\text{d}}{\rightarrow} \sigma^2_{\infty}(\bm{\theta})$ as $\tto$, where
\begin{equation*}
\sigma^2_{\infty}(\bm{\theta}) \overset{\text{d}}{=} \beta \int_0^{\infty} e^{-Y_{s}(\bm{\theta})} \diff s.
\end{equation*}
Setting the starting value as
\begin{equation}\label{eq:def:sig0}
\sigma^2_{0}(\bm{\theta}) \overset{\text{d}}{=}  \beta \int_0^{\infty} e^{-Y_{s}(\bm{\theta})} \diff s, \quad \text{independent of } L,
\end{equation}
by Theorem~3.2 of \cite{Kluppelberg04} for such $\bstheta$ the process $(\sigma_t^2(\bstheta))_{t \geq 0}$ is strictly stationary. 
Then by Proposition~4.2 of \cite{Kluppelberg04} for the stationary process and $k\in\N$,
\beam\label{eq:2.5}
\E\si_0^{2k}(\bstheta)<\infty\quad \mbox{if and only if}\quad \E L_1^{2k}<\infty\,\,\,\mbox{and}\,\,\,\Psi_{\bm{\theta}}(k)<0.
\eeam
Furthermore, for $k=1,2$ either of this implies that the squared returns from \eqref{eq:def:Gi} have corresponding finite moments (Proposition~5.1 of \cite{Kluppelberg04}). Additionally, by Corollary~3.1 of \cite{Kluppelberg04} the process $(P_t(\bstheta))_{t \geq 0}$ defined in \eqref{eq:def:Pt} with stationary $(\sigma_t(\bstheta))_{t\ge0}$ has stationary increments. 

\subsection{AR representation for the squared returns}

We estimate the COGARCH parameter, when the log-price process is observed on a regular grid of fixed size $\Delta>0$, such that the data are modelled by the returns $(G_i(\bm{\theta}))_{i \in \mathbb{N}}$ as defined in \eqref{eq:def:Gi}.

We state the basic assumptions and recall some properties of the COGARCH process. 

\begin{proposition}
[Theorems 3.1 and 3.4 in \citet{Haug07}]\label{th:haug}
Assume that:
\begin{itemize}
\item[(A1)] The parameter vector $\bm{\theta} = (\beta,\eta,\vp)^\top$ satisfies $\beta, \eta, \vp > 0$.
\item[(A2)] $\E L_1 = 0$ and $\var L_1 = 1$.
\item[(A3)] The variance $c_L$ of the Brownian component of L is known and satisfies $0 \leq c_L < \var L_1$.
\item[(A4)] $\E L_1^4 < \infty$.
\item[(A5)] $\int_{\mathbb{R}} x^3 \nu_L(\diff x) = 0$.
\item[(A6)] $\Psi_{\bm{\theta}}(2) < 0$.
\end{itemize}
Denote the expectation and variance of the squared returns process by 
\beao
\mu_{\bstheta} = \E  G_1^2(\bm{\theta})  \quad \text{and} \quad \gamma_{\bm{\theta}}(0) = \var G_1^2(\bm{\theta})
\eeao
Then the following holds:
\begin{itemize}
\item[(a)] 
The autocovariance function of the squared returns process is given by
\begin{equation}\label{eq:cov}
\gamma_{\bm{\theta}}(h) = \cov ( G_i^2(\bm{\theta}), G_{i+h}^2(\bm{\theta}) ) = 
\gamma_{\bm{\theta}}(0)k_{\bm{\theta}} e^{-h\rho_{\bm{\theta}}}, \quad h \in \mathbb{N}.
\end{equation}
\item[(b)] 
If  $\mu_{\bstheta},\gamma_{\bm{\theta}}(0),k_{\bm{\theta}},\rho_{\bm{\theta}} > 0$,  then these parameters uniquely determine $\bm{\theta}$.
\item[(c)]
 The process $(G_i(\bm{\theta}))_{i \in \mathbb{N}}$ is $\alpha$-mixing with exponentially decaying mixing coefficients.
\end{itemize}
\end{proposition}

Assume that the driving L\'{e}vy process satisfies assumptions (A2)-(A5) of Proposition~\ref{th:haug}. 
We take as parameter space of the COGARCH process a compact set $\Theta$ satisfying the relevant conditions of Proposition~\ref{th:haug}; more precisely,
\begin{equation}\label{eq:def:Th}
\Theta \subset \calm :=\{\bm{\theta} =(\beta,\eta,\vp)^\top : \beta,\eta,\vp > 0, \Psi_{\bm{\theta}}(2) < 0 \text{ and } \mu_{\bstheta},\gamma_{\bm{\theta}}(0),k_{\bm{\theta}},\rho_{\bm{\theta}} > 0\}.
\end{equation}
In what follows, we assume that the true model parameter $\bstheta_0 \in \Theta$.
We present the auxiliary AR model using the structure of COGARCH squared returns. 
Define the centered squared returns for $\bstheta\in\Theta$ as
\begin{equation}\label{eq:defG2Til}
\tilde{G}^2_i(\bm{\theta}) := G^2_i(\bm{\theta}) - \mu_{\bstheta}, \quad i \in \mathbb{N}.
\end{equation}

\begin{proposition}[Auxiliary AR$(r)$ model]\label{pr:cog:aux}
 Let $\bm{\theta} \in \Theta$ and $r \geq 2$ be fixed. Define
\begin{equation*}
U_i(\bm{\theta}) := \tilde{G}^2_{i+r}(\bm{\theta}) - P_{\calh_i}  \tilde{G}^2_{i+r}(\bm{\theta}), \quad i\in\N,
\end{equation*}
where  $\calh_i = {\rm \ov{sp}} \{\tilde{G}^2_{i+r-j}(\bm{\theta}),j=1,\dots,r\}$ is the closed span in the Hilbert space $\call^2$. 
Then  there exist unique real numbers $a_{\bm{\theta},1},\dots,a_{\bm{\theta},r}$ such that 
\begin{equation}\label{eq:2.8}
U_i(\bm{\theta}) = \tilde{G}^2_{i+r}(\bm{\theta}) - \sum_{j=1}^r a_{\bm{\theta},j} \tilde{G}^2_{i+r-j}(\bm{\theta}), \quad i \in \mathbb{N}.
\end{equation}
Moreover, the process $(U_i(\bm{\theta}))_{i \in \mathbb{N}}$ is strictly stationary with $\E U_i(\bm{\theta}) = 0$ and $\var U_i(\bm{\theta}) < \infty$.
\end{proposition}

\begin{proof}
The proof adapts the proof of Proposition~2.2 of \citet{Fasen18II} for the COGARCH process. 
Since $\bstheta \in \Theta \subset \calm$, 
by Proposition~\ref{th:haug}(a), the autocovariance function of $(\tilde{G}^2_i(\bm{\theta}))_{i \in \mathbb{N}}$ satisfies $\gamma_{\bm{\theta}}(0) > 0$ and $\gamma_{\bm{\theta}}(h) \rightarrow 0$ as $\nto$. 
By Proposition~5.1.1 of \citet{Brockwell13} it follows that the autocovariance matrix of $(\tilde{G}^2_i(\bm{\theta}))_{i =1}^r$ is non-singular. Hence, the numbers $a_{\bm{\theta},1},\dots,a_{\bm{\theta},r}$ are uniquely given by 
\begin{equation}\label{eq:yw}
\begin{pmatrix}
a_{\bm{\theta},1}  \\
a_{\bm{\theta},2} \\
\vdots \\
a_{\bm{\theta},r}
\end{pmatrix} = 
\begin{pmatrix}
\gamma_{\bstheta}(0) &  \gamma_{\bstheta}(1) & \dots  & \gamma_{\bstheta}(r-1) \\
\gamma_{\bstheta}(1) &  \gamma_{\bstheta}(0) & \dots  & \gamma_{\bstheta}(r-2) \\
\vdots & \vdots & & \vdots \\
\gamma_{\bstheta}(r-1) &  \gamma_{\bstheta}(r-2) & \dots  & \gamma_{\bstheta}(0) \\
\end{pmatrix}^{-1}
\begin{pmatrix}
\gamma_{\bstheta}(1)  \\
\gamma_{\bstheta}(2) \\
\vdots \\
\gamma_{\bstheta}(r)
\end{pmatrix}
\end{equation}
leading to \eqref{eq:2.8}.
\end{proof}

Proposition \ref{pr:cog:aux} gives an AR$(r)$ representation for $r\ge 2$ for the COGARCH squared returns from \eqref{eq:defG2Til} by rewriting \eqref{eq:2.8} as $\tilde{G}^2_{i+r}(\bm{\theta}) =  \sum_{j=1}^r a_{\bm{\theta},j} \tilde{G}^2_{i+r-j}(\bm{\theta}) + U_i(\bm{\theta})$ for $i\in\N$. 
Let
\begin{equation}\label{eq:def:pi:th}
\bspi_{\bstheta} := (\mu_{\bstheta},{\bm a}_\bstheta,\gamma_{\bm{\theta}}(0))^\top=(\mu_{\bstheta},a_{\bm{\theta},1},\dots,a_{\bm{\theta},r},\gamma_{\bm{\theta}}(0))^\top,
\end{equation}
and let $C \subset \R^r$ be a compact subset of the set containing all possible real coefficients of a strictly stationary AR$(r)$ process.
Then we define a compact parameter space of the auxiliary model as
\begin{equation}\label{eq:def:Pi}
\Pi := \Big[-\frac{1}{\epsilon},\frac{1}{\epsilon} \Big] \times C \times \Big[\epsilon,\frac{1}{\epsilon}\Big],
\end{equation}
where $\epsilon$ is a small positive constant.

We will investigate two well-known estimators of $\bspi_{\bstheta}$ in \eqref{eq:def:pi:th}, namely the least squares estimator (LSE) and the Yule-Walker estimator (YWE) defined by 
\begin{equation}\label{eq:def:pi_n}
\hat{\bspi}_{n,\ls}(\bstheta) = 
\begin{pmatrix}
\hat{\mu}_n(\bstheta) \\
\hat{\bm{a}}_{n,\ls}(\bstheta) \\
\hat{\gamma}_n(0;\bstheta)
\end{pmatrix} \quad \text{and} \quad
\hat{\bspi}_{n,\yw}(\bstheta) = 
\begin{pmatrix}
\hat{\mu}_n(\bstheta) \\
\hat{\bm{a}}_{n,\yw}(\bstheta) \\
\hat{\gamma}_n(0;\bstheta)
\end{pmatrix},
\end{equation}
respectively, whose components are given as follows.

%
%
\begin{definition}[LSE and YWE]\label{de:ywlse}
The estimators of the mean and variance are given by
\begin{equation*}
\hat{\mu}_n(\bstheta) = \frac{1}{n} \sum_{i=1}^n G_i^2(\bstheta) \quad \text{and} \quad \hat{\gamma}_n(0;\bstheta) = \frac{1}{n} \sum_{i=1}^n (G_i^2(\bstheta) - \hat{\mu}_n(\bstheta))^2.
\end{equation*}
(a) \, The LSE of $(a_{\bm{\theta},1},\dots,a_{\bm{\theta},r})^\top$ is given by
\begin{equation*}
\hat{\bm{a}}_{n,\ls}(\bstheta) =  \argminA_{\bm{c} \in C} S_n(\bm{c};\bstheta), 
\end{equation*}
for $C$ as in \eqref{eq:def:Pi}, and
\begin{equation*}
\begin{split}
& S_n(\bm{c};\bstheta) := \\
& \frac{1}{n-r} \sum_{i = 1}^{n-r} \Big( (G^2_{i+r}(\bstheta) - \hat{\mu}_n(\bstheta)) - c_1 (G^2_{i+r-1}(\bstheta) - \hat{\mu}_n(\bstheta)) - \dots - c_r(G^2_i(\bstheta) - \hat{\mu}_n(\bstheta))  \Big)^2.
\end{split}
\end{equation*}
%
%
(b) \, The YWE of $(a_{\bm{\theta},1},\dots,a_{\bm{\theta},r})^\top$ is given by
\begin{equation}\label{eq:3.6}
\hat{\bm{a}}_{n,\yw}(\bstheta) = 
\hat{\bm{\Gamma}}_{n}^{-1}(\bstheta)\hat{\bm{\gamma}}_{n}(\bstheta), \quad n \in \N,
\end{equation}
where $\hat{\bm{\Gamma}}_{n}^{-1}(\bstheta) = ( \hat{\gamma}_n(i-j;\bstheta) )_{i,j=1}^r$ and $\hat{\bm{\gamma}}_{n}(\bstheta) =(\hat{\gamma}_n(1;\bstheta),\dots,\hat{\gamma}_n(r;\bstheta))^\top$ are defined in terms of the empirical autocovariance function
\begin{equation*}
\hat{\gamma}_n(h;\bstheta) = \frac{1}{n} \sum_{i=1}^{n-h} (G^2_i(\bstheta) - \hat{\mu}_n(\bstheta))(G^2_{i+h}(\bstheta) - \hat{\mu}_n(\bstheta)), \quad h,n \in \N, n>h.
\end{equation*}
\end{definition}

We now define a function that will connect the COGARCH process to its auxiliary AR model from Proposition \ref{pr:cog:aux}.

\begin{proposition}[Binding function]\label{pr:bind:inj}
Define the binding function $\pi: \Theta \rightarrow \Pi$ by $\pi(\bm{\theta}) =  \bm{\pi_{\theta}}$ as in \eqref{eq:def:pi:th}. Then $\pi$ is injective and continuously differentiable for $r \geq 2$.
\end{proposition}

\begin{proof} 
As in the proof of Lemma 2.5 in \citet{Fasen18II}, we decompose $\pi: \Theta \rightarrow \Pi$ into three maps $\pi = \pi_1 \circ \pi_2 \circ \pi_3$. Define  $\pi_1: \Theta \rightarrow \mathbb{R}^4$ by  
\begin{equation*}
\pi_1(\bm{\theta}) = (\mu_{\bstheta}, k_{\bm{\theta}},\rho_{\bm{\theta}}, \gamma_{\bm{\theta}}(0))^\top,
\end{equation*}
which is by Proposition~\ref{th:haug}(b) injective.
Next define $\pi_2: \pi_1(\Theta) \rightarrow \mathbb{R}^{r+2}$ by
\begin{equation*}
\pi_2(\mu_{\bstheta}, k_{\bm{\theta}},\rho_{\bm{\theta}}, \gamma_{\bm{\theta}}(0)) = (\mu_{\bstheta},\gamma_{\bm{\theta}}(1),\dots,\gamma_{\bm{\theta}}(r), \gamma_{\bm{\theta}}(0))^\top.
\end{equation*}
By \eqref{eq:cov}, $\gamma_{\bm{\theta}}(h) = \gamma_{\bm{\theta}}(0)k_{\bm{\theta}} e^{-h\rho_{\bm{\theta}}}$ for every $h \in \mathbb{N}$, and simple algebra shows that $k_{\bm{\theta}}$ and $\rho_{\bm{\theta}}$ are uniquely determined by
\begin{equation}\label{eq:pr:bi1}
k_{\bm{\theta}} = \frac{\gamma^2_{\bm{\theta}}(1)}{\gamma_{\bm{\theta}}(0)\gamma_{\bm{\theta}}(2)} \quad \text{ and } \quad \rho_{\bm{\theta}} = \log \Big(\frac{\gamma_{\bm{\theta}}(1)}{\gamma_{\bm{\theta}}(2)}\Big),
\end{equation}
and, therefore, $\pi_2$ is injective. Finally, define the map $\pi_3: \pi_2(\pi_1(\Theta)) \rightarrow \Pi$, by 
\begin{equation*}
\pi_3(\mu_{\bstheta},\gamma_{\bm{\theta}}(1),\dots,\gamma_{\bm{\theta}}(r), \gamma_{\bm{\theta}}(0)) = (\mu_{\bstheta},a_{\bm{\theta},1},\dots,a_{\bm{\theta},r},\gamma_{\bm{\theta}}(0))^\top.
\end{equation*}
The map $\pi_3$ is injective, since $\gamma_{\bm{\theta}}(1),\dots,\gamma_{\bm{\theta}}(r)$ are uniquely determined by $a_{\bm{\theta},1},\dots,a_{\bm{\theta},r}$ and $\gamma_{\bm{\theta}}(0)$. 
We need $r \geq 2$ in order to recover  $(\gamma_{\bm{\theta}}(1),\gamma_{\bm{\theta}}(2))$ from $(\gamma_{\bm{\theta}}(0),a_{\bm{\theta},1},a_{\bm{\theta},2})$ using the system of Yule-Walker equations \eqref{eq:yw}, so that \eqref{eq:pr:bi1} remains valid. This implies the injectivity of the composition $\pi$.

Now we prove that $\pi$ is continuously differentiable. 
The map $\pi_1$ is given in terms of equations (3.6)-(3.9) of Theorem~3.1 in \citet{Haug07}, which are continuously differentiable maps of $\Psi_{\bm{\theta}}(1)$ and $\Psi_{\bm{\theta}}(2)$ as defined in \eqref{eq:mgf:yt}. By assumption (A4) of Proposition~\ref{th:haug} the L\'evy process $L$ has finite fourth moment and, therefore, both $\Psi_{\bm{\theta}}(1)$ and $\Psi_{\bm{\theta}}(2)$ exist and are continuously differentiable in $\bstheta$. 
By \eqref{eq:cov}, $\pi_2$ is  continuously differentiable. 
Finally, $\pi_3$ is also continuously differentiable since it is defined recursively by means of the Yule-Walker equations \eqref{eq:yw}. This proves that the composition $\pi$ is continuously differentiable.
\end{proof}

\subsection{Indirect Inference Estimation}



Let $\hatbspi_n(\bstheta)$ denote an estimator of the auxiliary AR$(r)$ model for $r \geq 2$ based on the returns $(G_i(\bstheta))_{i=1}^n$, where $\bm{\theta}$ lies in a compact subset $\Theta$ of  $\calm$ as in \eqref{eq:def:Th}. 
We define now the IIE for the COGARCH process.

\begin{definition}[IIE]\label{de:ii}
Let $\bm{G}_n:=(G_i(\bm{\theta}_0))_{i=1}^n$ be  the  returns as defined in \eqref{eq:def:Gi}.
Let $\hatbspi_n$ be one of the estimators given in \eqref{eq:def:pi_n} of $\bm{\pi}_{\bm{\theta}_0}$ as defined in \eqref{eq:def:pi:th}.
{For arbitrary $\bstheta \in \Theta$ and $k = 1,\dots,K$} let $\hatbspi_{n,k}(\bm{\theta})$ be estimators of $\bm{\pi_{\theta}}$ based on independent simulated paths $\bm{G}_{n,k}(\bm{\theta}):=(G_i^{(k)}(\bm{\theta}))_{i=1}^n$.
Let $\bm{\Omega}$ be a symmetric and positive definite weight matrix. 
Define the function
\begin{equation*}
\hat{L}_{\ii} : \Theta \rightarrow [0,\infty) \quad \text{based on $\bm{G}_n$ by}\quad
\hat{L}_{\ii}(\bm{\theta},\bm{G}_n):= \Big\|\hatbspi_n - \frac{1}{K} \sum_{k=1}^K \hatbspi_{n,k}(\bstheta) \Big \|_{\bm{\Omega}}.
\end{equation*}
Then the IIE of $\bm{\theta}_0$ is defined as 
\begin{equation}\label{eq:DefIIE}
\hat{\bm{\theta}}_{n,\ii} := \argminA_{\bm{\theta} \in \Theta} \hat{L}_{\ii}(\bm{\theta},\bm{G}_n).
\end{equation}
\end{definition}
Concerning the asymptotic behavior of the IIE one would hope that strong consistency and asymptotic normality of the auxiliary model estimators also implies strong consistency and asymptotic normality of the IIE. 
However, as the Indirect Inference method is based on the simulation of the COGARCH process for many different parameters,
we need a stronger (uniform) consistency result and also additional regularity conditions to ensure  this. 
The following is a modification of Propositions~1 and~3 of \citet{Gourieroux93}, and it is the analog of Theorem~3.2 of \citet{Fasen18II} in the context of our model.

\begin{proposition}
\label{th:ii:asymp} 
Assume the setting of Definition \ref{de:ii} and $r \geq 2$.
\begin{itemize}
\item[(a)] 
If the uniform SLLN 
\begin{equation}\label{eq:2.15}
\sup_{\bm{\theta} \in \Theta}  \| \hatbspi_n(\bm{\theta}) - \bm{\pi}_{\bstheta} \| \overset{\text{a.s.}}{\rightarrow} 0, \quad \nto,
\end{equation}
holds, then the IIE \eqref{eq:DefIIE} is strongly consistent:
\begin{equation*}
\hat{\bm{\theta}}_{n,\ii} \overset{\text{a.s.}}{\rightarrow} \bm{\theta}_0, \quad \nto.
\end{equation*}
\item[(b)] 
Assume additionally to \eqref{eq:2.15} that the following hold:
\begin{itemize}
\item[(b.1)] 
for every $n \in \N$ the map $\bstheta \mapsto \hat{\bspi}_n(\bstheta)$ is continuously differentiable,
\item[(b.2)] for every $\bstheta \in \Theta$ we have $\sqrt{n}(\hat{\bspi}_n(\bstheta) - \bspi_{\bstheta}) \std \mathcal{N}(0,\bm{\Sigma}_{\bstheta})$ as $\nto$, and
\item[(b.3)] 
for every sequence $(\bstheta_n)_{n \in \N}$ with $\bstheta_n \stas \bstheta_0$ it also holds that
\begin{equation*}\label{eq:cond:b3}
\nabla_{\bstheta} \hat{\bspi}_n(\bstheta_n) \stp  \nabla_{\bstheta} \bspi_{\bstheta_0},\quad \nto,
\end{equation*}
and $\gradtheta \bspi(\bstheta_0)$ has full column rank 3.
\item[(b.4)] The true parameter $\bstheta_0$ lies in the interior of $\Theta$.
\end{itemize}
Then the IIE \eqref{eq:2.15} is asymptotically normal:
\begin{equation*}
\sqrt{n}(\hat{\bm{\theta}}_{n,\ii} - \bstheta_0) \std \mathcal{N}(0,\Xi_{\bstheta_0}), \quad \nto,
\end{equation*}
where the asymptotic variance is given by
\begin{equation}\label{eq:2.19}
\Xi_{\bstheta_0} = (\mathcal{J}_{\bstheta_0})^{-1}\mathcal{I}_{\bstheta_0}(\mathcal{J}_{\bstheta_0})^{-1}
\end{equation}
with
\beam
\mathcal{J}_{\bstheta_0} &=& (\nabla_{\bstheta} \pi_{\bstheta_0})^\top \bm{\Omega} (\nabla_{\bstheta} \pi_{\bstheta_0}) \mbox{ and}\nonumber\\
\mathcal{I}_{\bstheta_0} &=& (\nabla_{\bstheta} \pi_{\bstheta_0})^\top\bm{\Omega}\Big(1 + \frac{1}{K}\Big)\bm{\Sigma}_{\bstheta}\bm{\Omega}(\nabla_{\bstheta} \pi_{\bstheta_0}).\label{eq:2.22}
\eeam
\end{itemize}
\end{proposition}

\begin{proof}
Part (a) follows as a particular case of the proof of Theorem~3.2 of \cite{Fasen18II}. For part (b), we need to check assumptions (C.3)-(C.5) of Theorem~3.2 in \cite{Fasen18II}. By construction of the estimator \eqref{eq:def:pi_n} the asymptotic covariance matrices in (C.3) and (C.4) are identical, so that (b.2) implies (C.3) and (C.4). Instead of verifying (C.5) we modify the argument in \cite{Fasen18II} (under (b.1) and (b.4)) when manipulating the first order condition
\begin{equation}\label{eq:foc}
0 = \gradtheta \hat{L}_{\ii}(\hat{\bm{\theta}}_{n,\ii},\bm{G}_n) = 2 (\gradtheta \hatbspi_{n}(\hat{\bm{\theta}}_{n,\ii}))^T \Omega (\hatbspi_{n}(\hat{\bm{\theta}}_{n,\ii}) - \hatbspi_n).
\end{equation}
We perform, as in Theorem~3.2 in \citet{Newey94Large}, a Taylor expansion of order 1 around the true value $\bstheta_0$ of the function $\hatbspi_{n}(\hat{\bm{\theta}}_{n,\ii})$ in \eqref{eq:foc}. After rearranging the terms this leads to
\begin{equation*}
\sqrt{n}(\hat{\bm{\theta}}_{n,\ii} - \bstheta_0) = 
-\Big( (\gradtheta \hatbspi_{n}(\hat{\bm{\theta}}_{n,\ii}))^T \Omega (\gradtheta \hatbspi_{n}(\bstheta_n)) \Big)^{-1} (\gradtheta \hatbspi_{n}(\hat{\bm{\bstheta}}_{n,\ii})) \Omega \sqrt{n} (\hatbspi_{n}(\bstheta_0) - \hatbspi_n),
\end{equation*}
where $\bstheta_n$ is such that $\|\bstheta_n - \bstheta_0\| \leq \| \hat{\bm{\bstheta}}_{n,\ii} - \bstheta_0\|$. The asymptotic normality is now a consequence of taking the limit for $\nto$ and using (a), (b.2) and (b.3).
\end{proof}

%
%

\section{Auxiliary AR model - strong consistency and asymptotic normality}

Our objective is to investigate the asymptotic behavior of the  IIE for the COGARCH parameter $\bstheta$ using an AR$(r)$ model for fixed $r \geq 2$ as auxiliary model. 
This amounts to verifying all assumptions of Proposition~\ref{th:ii:asymp}. 
As a first step we investigate consistency and joint asymptotic normality of the parameter estimator of the auxiliary AR$(r)$ model, which results from the projection of Proposition~\ref{pr:cog:aux}, and may have a non-zero mean.
It is worth noticing that the noise $(U_i(\bstheta))_{i \in \N}$ from \eqref{eq:2.8} is defined as the projection error over the finite past.
Thus, $U_i(\bstheta)$ is orthogonal to $\calh_i = {\rm \ov{sp}} \{\tilde{G}^2_{i+r-j}(\bm{\theta}),j=1,\dots,r\}$, but we cannot guarantee that it is also orthogonal to ${\rm \ov{sp}} \{\tilde{G}^2_{i+r-j}(\bm{\theta}),j \in \N \}$, so it may not be a white noise process.  Therefore, the classical asymptotic theory for the estimation of autoregressive processes (when data come from an AR model with white noise residuals) does not apply directly.
Since the residuals are stationary and ergodic with zero mean and finite variance, and since  $U_i(\bstheta)$ is orthogonal to $\calh_i = {\rm \ov{sp}} \{\tilde{G}^2_{i+r-j}(\bm{\theta}),j=1,\dots,r\}$, we can obtain results by modifying the classical arguments.

We shall do this for the two estimators from Definition~\ref{de:ywlse} and recall that in classical time series theory they are asymptotically equivalent (cf. the proof of Theorem 8.1.2 in \citet{Brockwell13}). 
To the best of our knowledge, this has not yet been covered in the literature.



\begin{remark}
This section provides asymptotic results for the estimators of the auxiliary AR model for some arbitrary, but fixed COGARCH parameter $\bstheta$, where the dependence on $\bstheta$ is irrelevant, and we omit it for ease of notation.
We define $(W_i)_{i \in \N} := (G_i^2)_{i \in \N}$ and rewrite the auxiliary AR$(r)$ model of Proposition~\ref{pr:cog:aux} with parameter $\bspi = (\mu,\bm{a},\gamma(0))$ as
\beao
\tilde{W}_{i+r}  =  \sum_{j = 1}^r a_j \tilde{W}_{i+r-j} + U_i,\quad i\in\N,
\eeao
where $\tilde{W}_i =\tilde G_i^2 = W_i - \mu, \mu = \E W_1$ and $\gamma(0) = \var W_1$. 
\end{remark}


\subsection{Strong consistency of LSE and YWE}

\begin{lemma}\label{le:mu:cov}
Let the assumptions of Proposition~\ref{th:haug} hold. Then as $\nto$, $\hat{\mu}_n \stas \mu$ and $\hat{\gamma}_n(h) \stas \gamma(h)$ for all $h \in \N_0$.
\end{lemma}

\begin{proof}
From Proposition~\ref{th:haug} we know that $\E |W_1| < \infty$ and $(W_i)_{i \in \N}$ is ergodic, so that Birkhoff's ergodic theorem (see e.g. Theorem 4.4 in \citet{Krengel85}) gives immediately $\hat{\mu}_n \stas \mu$ as $\nto$. 
To prove almost sure convergence of the empirical autocovariance function, we first investigate it, when the mean $\mu$ is known:
\begin{equation}\label{eq:gmStar}
\gamma^{\ast}_n(h) := \frac{1}{n} \sum_{i=1}^{n-h} (W_i - \mu)(W_{i+h} - \mu), \quad h \in \N_0.
\end{equation}
{Since $W_iW_{i+h}$ is for every $i\in\N$ a measurable map of finitely many values of $(W_i)_{i \in \N}$,
the sequence $(W_iW_{i+h})_{i \in \N}$ is ergodic.}
From Proposition~\ref{th:haug}(a), $\E|W_1W_{1+h}| < \infty$, so that Birkhoff's ergodic theorem gives $\gamma^{\ast}_n(h) \stas \gamma(h)$ as $\nto$. 
Simple algebra shows that
\begin{equation}\label{eq:difCorr}
\hat{\gamma}(h) - \gamma^{\ast}(h) = \frac{1}{n}\sum_{i=1}^{n-h} (W_i + W_{i+h} - \hat{\mu}_n -  \mu )(\mu -\hat{\mu}_n).
\end{equation}
Since $\hat{\mu}_n \stas \mu$, the difference $\gamma^{\ast}_n(h) - \hat{\gamma}_n(h)\stas 0$ as $\nto$;
hence, $\hat{\gamma}_n(h) \stas \gamma(h)$ as $\nto$.
\end{proof}

\begin{theorem}[Consistency of LSE and YWE]\label{th:lse:con}
Let the assumptions of Proposition~\ref{th:haug} hold. Then as $\nto$, $\hat\bsa_{n,\ls} \stas \bm{a}$ and $\hat\bsa_{n,\yw} \stas \bm{a}$.
\end{theorem}

\begin{proof}
We start by proving strong consistency of the LSE, when the mean $\mu$ is known:
\begin{equation}\label{eq:lseStar}
\bsa_{n,\ls}^{\ast} =  \argminA_{\bm{c} \in C}S^{\ast}_n(\bm{c}),
\end{equation}
for $C$ as in \eqref{eq:def:Pi}, and
\begin{equation*}
S^{\ast}_n(\bm{c}) = \frac{1}{n-r} \sum_{i = 1}^{n-r} \big( (W_{i+r} - \mu) - c_r (W_{i+r-1} - \mu) - \dots - c_1(W_i - \mu)\big)^2.
\end{equation*}
As in Section 8.10* of \cite{Brockwell13} we write the auxiliary AR$(r)$ model in matrix form as
\begin{equation*}
 \tildebsy = \tildebsw \bm{a} + \bm{U}_n, \quad n \in \N,
\end{equation*}
where $\tildebsy = (\tilde{W}_{r+1},\dots,\tilde{W}_n)^\top, \bm{U}_n = (U_1,\dots,U_{n-r})^\top$ and $\tildebsw$ is the $n \times r$ design matrix,
\begin{equation}\label{eq:matr:ar}
\tildebsw = \begin{pmatrix}
\tilde{W}_r & \tilde{W}_{r-1} & \dots &  \tilde{W}_{1} \\
\tilde{W}_{r+1} & \tilde{W}_{r} & \dots &  \tilde{W}_{2} \\
\vdots & \vdots &  &  \vdots \\
\tilde{W}_{n-1}& \tilde{W}_{n-2} & \dots &  \tilde{W}_{n-r}
\end{pmatrix}.
\end{equation}
Then notice that
\begin{equation*}
(n-r) S^{\ast}_n(\bm{c}) = (\tildebsy - \tildebsw \bm{c})^\top (\tildebsy - \tildebsw \bm{c}),
\end{equation*}
revealing the LSE as a linear regression-type estimator given by
\begin{equation}\label{eq:lse:reg}
\bsa_{n,\ls}^{\ast} = (\tildebsw^\top\tildebsw)^{-1}\tildebsw^\top\tildebsy,
\end{equation}
provided that the $r \times r$ matrix $\tildebsw^\top\tildebsw$ is invertible. We prove that $n^{-1}\tildebsw^\top\tildebsw$ converges a.s. to an invertible matrix. For each fixed $u,v \in \{1,\dots,r\}$ the $(u,v)$-th entry of this matrix is 
\begin{equation*}
\frac{1}{n}\sum_{i = 0}^{n-r-1} \tilde{W}_{r+1-u+i}\tilde{W}_{r+1-v+i}. 
\end{equation*}
Since $\tilde{W}_{r+1-u+i}\tilde{W}_{r+1-v+i}$ is for every $i \in \N_0$ a measurable map of finitely many values of  $(W_i)_{i \in \N}$, {the sequence $(\tilde{W}_{r+1-u+i}\tilde{W}_{r+1-v+i})_{i\in\N_0}$ is ergodic.} Since $\E W_{1}^2 < \infty$ Birkhoff's ergodic theorem gives
\begin{equation}\label{eq:wTw_lim}
\frac{1}{n}\sum_{i = 0}^{n-r-1} \tilde{W}_{r+1-u+i}\tilde{W}_{r+1-v+i} \stas \E \tilde{W}_1 \tilde{W}_{1 + |u-v|}, \quad \nto,
\end{equation}
and thus $n^{-1}\tildebsw^\top\tildebsw \stas \bm{\Gamma}$ as $\nto$, where $\bm{\Gamma}$ is the autocovariance matrix of the squared COGARCH returns, which is non-singular (cf. the proof of Proposition~\ref{pr:cog:aux}). 
Thus, $\bm{\Gamma}$ is invertible and therefore the estimator given in \eqref{eq:lse:reg} is well defined for $n$ large enough. 
With \eqref{eq:lse:reg} we calculate
\beam
\bsa_{n,\ls}^{\ast} - \bm{a} & = & (\tildebsw^\top\tildebsw)^{-1}\tildebsw^\top(\tildebsw \bm{a} + \bm{U}_n) - \bm{a} \nonumber\\
& = & n(\tildebsw^\top\tildebsw)^{-1} \frac{1}{n} \tildebsw^\top\bm{U}_n\nonumber \\
& = & n(\tildebsw^\top\tildebsw)^{-1} \frac{1}{n} \sum_{i=1}^n \Big( \tilde{W}_{i+r} - \sum_{j=1}^r a_j \tilde{W}_{i+r-j} \Big)
\begin{pmatrix}
\tilde{W}_{i+r-1}  \\
\vdots \\
\tilde{W}_{i}
\end{pmatrix} \nonumber\\ 
& =: & (n^{-1}\tildebsw^\top\tildebsw)^{-1} \frac{1}{n} \sum_{i=1}^n \bm{Z}_i. \label{eq:3.18}
\eeam
Since $\bm{Z}_i$ is for every $i \in \N$ a measurable map of finitely many values of  $(W_i)_{i \in \N}$, {the sequence $(\bm{Z}_i)_{i\in\N}$ is ergodic.} According to Proposition~\ref{pr:cog:aux}, $\tilde{W}_{i+r} - \sum_{j=1}^r a_j \tilde{W}_{i+r-j}$ is uncorrelated with $\tilde{W}_i,\dots,\tilde{W}_{i+r-1}$ for all $i \in \N$.
Since $\E |\bm{Z}_1| < \infty$ Birkhoff's ergodic theorem gives 
\begin{equation*}
\frac{1}{n} \sum_{i=1}^n \bm{Z}_i \stas \E \bm{Z}_1 = \bm{0}.
\end{equation*}
This together with the fact that the first term of \eqref{eq:3.18} converges a.s. to $\bm{\Gamma}^{-1}$ shows that
$$
\bsa_{n,\ls}^{\ast} \stas \bm{a}, \quad \nto.
$$
It remains to prove that $(\hat{\bm{a}}_{n,\ls} - \bsa_{n,\ls}^{\ast}) \stas \bm{0}$ as $\nto$. 
Write the LSE in the matrix form
\beao\label{eq:3.8}
\hat{\bm{a}}_{n,\ls} = (\barbsw^\top\barbsw)^{-1}\barbsw^\top\barbsy,
\eeao
where $\barbsw$ and $\barbsy$ denote the matrix and vector defined in Eq.~\eqref{eq:matr:ar}, with entries of the form $\bar{W}_i = W_i - \hat{\mu}_n$. 
Using the matrix identity $A^{-1}x - C^{-1}y = A^{-1}(x-y) + A^{-1}(C-A)C^{-1}y$ gives
\begin{equation}\label{eq:lsSqSta}
\begin{split}
 (\hat{\bsa}_{n,\ls} - \bsa_{n,\ls}^{\ast}) & = \Big(\frac{\barbsw^\top\barbsw}{n}\Big)^{-1}\Big(\frac{\barbsw^\top\barbsy}{n} - \frac{\tildebsw^\top\tildebsy}{n}\Big) \\
& + \Big(\frac{\barbsw^\top\barbsw}{n}\Big)^{-1}\Big( \frac{\tildebsw^\top\tildebsw}{n} - \frac{\barbsw^\top\barbsw}{n}\Big)\Big(\frac{\tildebsw^\top\tildebsw}{n}\Big)^{-1}\Big(\frac{\tildebsw^\top\tildebsy}{n}\Big).
\end{split}
\end{equation}
By Birkhoff's ergodic theorem $n^{-1}\barbsw^\top\barbsw, n^{-1}\tildebsw^\top\tildebsw$ and $n^{-1}\tildebsw^\top\tildebsy$ converge a.s. to two matrices and a vector, respectively. 
Additionally, by \eqref{eq:difCorr} we can apply Birkhoff's ergodic theorem to obtain as $\nto$,
\begin{equation*}
\Big(\frac{\barbsw^\top\barbsy}{n} - \frac{\tildebsw^\top\tildebsy}{n}\Big) \stas \bm{0} \quad \text{and} \quad \Big( \frac{\tildebsw^\top\tildebsw}{n} - \frac{\barbsw^\top\barbsw}{n}\Big) \stas \bm{0},
\end{equation*}
showing that the LSE is consistent. For the YWE the proof is a direct consequence of Lemma~\ref{le:mu:cov} and the continuous mapping theorem.
\end{proof}

\subsection{Asymptotic normality of the LSE and YWE}

One of the requirements for asymptotic normality of the IIE of the COGARCH parameter $\bstheta$ is condition (b.2) of Proposition~\ref{th:ii:asymp}. 
This means we have to prove asymptotic normality of $\hat{\bspi}_{n,\ls}$ and $\hat{\bspi}_{n,\yw}$. 
We start with an auxiliary result.

\begin{lemma}\label{le:ywlse:st}
Let the assumptions of Proposition~\ref{th:haug} hold. Let $\bsa_{n,\ls}^{\ast}$ be the LSE defined in \eqref{eq:lseStar} and $\bsa_{n,\yw}^{\ast}$ be the modification of the YWE defined in \eqref{eq:3.6},  when the true mean $\mu$ is known, i.e., with $\hat{\gamma}_n(\cdot)$ replaced by $\gamma^{\ast}(\cdot)$ from \eqref{eq:gmStar}. Then as $\nto$,
\begin{itemize}
\item[(a)] $\sqrt{n}( \hat{\mu}_n^2 - \mu^2) \stp 0$,
\item[(b)] $\sqrt{n}( \bsa_{n,\yw}^{\ast} - \hat{\bm{a}}_{n,\yw} ) \stp 0,$
\item[(c)] $\sqrt{n}( \bsa_{n,\yw}^{\ast} - \bsa_{n,\ls}^{\ast} ) \stp 0,$ and
\item[(d)] $\sqrt{n}( \bsa_{n,\ls}^{\ast} - \hat{\bm{a}}_{n,\ls}) \stp 0.$
\end{itemize}
\end{lemma}

\begin{proof}
(a) \, Write $\sqrt{n}( \hat{\mu}_n^2 - \mu^2) =  \sqrt{n} (\hat{\mu}_n +  \mu)(\hat{\mu}_n -  \mu)$ and notice that by Lemma~\ref{le:mu:cov} we only need to show that $\sqrt{n} (\hat{\mu}_n +  \mu)$ is bounded in probability. It follows from \eqref{eq:cov} that $\gamma(h)$ decays exponentially with $h$ and thus $\sum_{h=-\infty}^{\infty} |\gamma(h)| < \infty$. Let $\epsilon > 0$ be fixed and apply Chebyshev's inequality to get
$$
\P(\sqrt{n}|\hat{\mu}_n +  \mu| > \epsilon ) \leq \epsilon^{-2} n \var(\hat{\mu}_n) \rightarrow \epsilon^{-2} \sum_{h=-\infty}^{\infty} \gamma(h)<\infty, \quad \nto,
$$
where the convergence follows from Theorem~7.1.1 in \citet{Brockwell13}. 
\\[2mm]
(b) \, Write $\bsa_{n,\yw}^{\ast} = (\bm{\Gamma}_n^{\ast})^{-1}\bm{\gamma}_n^{\ast}$ with autocovariance function $\gamma^{\ast}(\cdot)$ defined in \eqref{eq:gmStar}. Using properties of the inverse matrix we get
\begin{equation*}
\begin{split}
& \quad \,\,\sqrt{n}( \bsa_{n,\yw}^{\ast} - \hat{\bm{a}}_{n,{\yw}} ) \\
& = \sqrt{n}\big(\hat{\bm{\Gamma}}_n^{-1}\hat{\bm{\gamma}}_n - (\bm{\Gamma}_n^{\ast})^{-1}\bm{\gamma}_n^{\ast}\big) \\
& = \hat{\bm{\Gamma}}_n^{-1}\sqrt{n} (\bm{\Gamma}_n^{\ast}-\hat{\bm{\Gamma}}_n)(\bm{\Gamma}_n^{\ast})^{-1}\hat{\bm{\gamma}}_n + (\bm{\Gamma}_n^{\ast})^{-1}\sqrt{n}(\hat{\bm{\gamma}}_n - \bm{\gamma}_n^{\ast}).
\end{split}
\end{equation*}
The estimators $\hat{\bm{\Gamma}}_n, \bm{\Gamma}_n^{\ast}$ and $\hat{\bm{\gamma}}_n$ are all bounded in probability. For fixed $h \in \N_0$ it follows from \eqref{eq:difCorr} and Lemma~\ref{le:ywlse:st}(a) that $\sqrt{n}(\hat{\gamma}_n(h) - \gamma^{\ast}_n(h)) \stp 0$ as $\nto$. Therefore,  $\sqrt{n}(\bm{\Gamma}_n^{\ast}-\hat{\bm{\Gamma}}_n)$ and $\sqrt{n}(\hat{\bm{\gamma}}_n - \bm{\gamma}_n^{\ast})$ also converge to zero in probability, which entails (b).\\[2mm]
(c) \, This follows similarly as in the proof of Theorem~8.1.1 in \citet{Brockwell13}.\\[2mm]
(d) \,  By \eqref{eq:lsSqSta} and observing that $n^{-1}\barbsw^\top\barbsw$, $n^{-1}\tildebsw^\top\tildebsw$ and $n^{-1}\tildebsw^\top\tildebsy$ are bounded in probability, we only need to show that $n^{-\frac{1}{2}}\{\barbsw^\top\barbsy - \tildebsw^\top\tildebsy\}$ and $n^{-\frac{1}{2}}\{\tildebsw^\top\tildebsw - \barbsw^\top\barbsw\}$ converge to zero in probability as $\nto$. These terms only depend on the autocovariance function of the process $(W_i)_{i \in \N}$ and therefore convergence in probability to zero follows from $\sqrt{n}(\hat{\gamma}_n(h) - \gamma^{\ast}_n(h)) \stp 0$ as $\nto$ as can be seen from \eqref{eq:difCorr}, and the fact that $\sqrt{n}\hat{\mu}_n$ is bounded in probability. 
\end{proof}

The following is the main result of this Section and proves Proposition~\ref{th:ii:asymp}(b.2).

\begin{theorem}[Asymptotic normality of the LSE and YWE]\label{th:ar:norm}
Let the assumptions of Proposition~\ref{th:haug} hold. Assume additionally that $\E |L_1|^{8+\epsilon} < \infty$ and $\Psi_{\bstheta}(4+\frac{\epsilon}{2}) < 0$ for some $\epsilon > 0$ and that the matrix $\bm{\Sigma}$ defined  in \eqref{eq:defSig} is  positive definite.
Then, both LSE and YWE for the AR$(r)$ model for $r\ge 2$ are asymptotically normal with covariance matrix
\begin{equation}\label{eq:defSig}
\bm{\Sigma} = \begin{blockarray}{c|ccc|c}
    \begin{block}{(c|ccc|c@{\hspace*{5pt}})}
    1 & 0 & \dots & 0  & 0 \\
        \cline{1-5}
    0 & \BAmulticolumn{3}{c|}{\multirow{3}{*}{$\bm{\Gamma}^{-1}$}} & 0 \\
    \vdots & &&& \vdots \\
    0 & &&& 0 \\
    \cline{1-5}
    0 & 0& \dots & 0 & 1 \\
    \end{block}
  \end{blockarray} \,\,\, \bm{\Sigma}^{\ast},
\end{equation}
where $\bm{\Gamma}$ is the autocovariance matrix of $(W_i)_{i=1}^r$,
\begin{equation}\label{eq:defSigAs}
\bm{\Sigma}^{\ast} = \E \bm{C}_1 \bm{C}_1^\top + 2\sum_{i=1}^{\infty} \E  \bm{C}_1 \bm{C}_{1+i}^\top,
\end{equation}
with $C_i\in\R^{r+2}$ given by
\beam\label{cmatrix}
\bm{C}_i = \begin{pmatrix}
\tilde{W}_i \\
 ( \tilde{W}_{i+r} - \sum_{j=1}^r a_j \tilde{W}_{i+r-j} )\tilde{W}_{i+r-1} \\
\vdots \\
 ( \tilde{W}_{i+r} - \sum_{j=1}^r a_j \tilde{W}_{i+r-j} )\tilde{W}_{i} \\
W_i^2 - \mu^2
\end{pmatrix}
\eeam
\end{theorem}

\begin{proof}
Write
\begin{equation}\label{eq:norm:1}
\sqrt{n}(\hatbspi_{n,\ls} - \bspi) = \sqrt{n}\begin{pmatrix}
\hat{\mu}_n  - \mu \\
\hat{\bm{a}}_{n,\ls} - \bm{a} \\
\hat{\gamma}_n(0) - \gamma(0)
\end{pmatrix} = 
\sqrt{n}\begin{pmatrix}
0 \\
\hat{\bm{a}}_{n,\ls} - \bsa_{n,\ls}^{\ast} \\
\mu^2 - \hat{\mu}_n^2
\end{pmatrix} + 
\sqrt{n}\begin{pmatrix}
\hat{\mu}_n  - \mu \\
\bsa_{n,\ls}^{\ast} - \bm{a} \\
\hat{\gamma}_n(0) + \hat{\mu}_n^2 - \E W_1^2
\end{pmatrix}
\end{equation}
and
\begin{equation*}
\hspace{-9cm}
\sqrt{n}(\hatbspi_{n,\yw} - \bspi)  =
\sqrt{n}\begin{pmatrix}
\hat{\mu}_n  - \mu \\
\hat{\bm{a}}_{n,\yw} - \bm{a} \\
\hat{\gamma}_n(0) - \gamma(0)
\end{pmatrix} \\
\end{equation*}
\begin{equation}\label{eq:norm:2}
= \sqrt{n}\begin{pmatrix}
0 \\
\hat{\bm{a}}_{n,\yw} - \bsa_{n,\yw}^{\ast} \\
\mu^2 - \hat{\mu}_n^2
\end{pmatrix} + 
\sqrt{n}\begin{pmatrix}
0 \\
 \bsa_{n,\yw}^{\ast} - \bsa_{n,\ls}^{\ast} \\
0
\end{pmatrix} + 
\sqrt{n}\begin{pmatrix}
\hat{\mu}_n  - \mu \\
\bsa_{n,\ls}^{\ast} - \bm{a} \\
\hat{\gamma}_n(0) + \hat{\mu}_n^2 - \E W_1^2
\end{pmatrix}.
\end{equation}
We apply Lemma~\ref{le:ywlse:st} to the right-hand side of \eqref{eq:norm:1} and \eqref{eq:norm:2} and find that it suffices to prove that
\begin{equation*}
\sqrt{n}\begin{pmatrix}
\hat{\mu}_n  - \mu \\
\bsa_{n,\ls}^{\ast} - \bm{a} \\
\hat{\gamma}_n(0) + \hat{\mu}_n^2 - \E W_1^2
\end{pmatrix} \std \mathcal{N}(0,\bm{\Sigma}), \quad \nto.
\end{equation*}
Using \eqref{eq:3.18} we write
\begin{equation}\label{eq:def:bn:ci}
\begin{split}
 \sqrt{n}\begin{pmatrix}
\hat{\mu}_n  - \mu \\
\bsa_{n,\ls}^{\ast} - \bm{a} \\
\hat{\gamma}_n(0) + \hat{\mu}_n^2 - \E W_1^2
\end{pmatrix} & =  
\sqrt{n}\begin{pmatrix}
 \frac{1}{n} \sum_{i=1}^n (W_i  - \mu) \\
n(\tildebsw^\top\tildebsw)^{-1} \frac{1}{n} \sum_{i=1}^n \bm{Z}_i \\
 \frac{1}{n} \sum_{i=1}^n (W_i - \hat{\mu}_n)^2 + \hat{\mu}_n^2 - \E W_1^2 \\
\end{pmatrix}  \\
& = \begin{blockarray}{c|ccc|c}
    \begin{block}{(c|ccc|c@{\hspace*{5pt}})}
    1 & 0 & \dots & 0  & 0 \\
        \cline{1-5}
    0 & \BAmulticolumn{3}{c|}{\multirow{3}{*}{$n(\tildebsw^\top\tildebsw)^{-1}$}} & 0 \\
    \vdots & &&& \vdots \\
    0 & &&& 0 \\
    \cline{1-5}
    0 & 0& \dots & 0 & 1 \\
    \end{block}
  \end{blockarray} \,\, \frac{1}{\sqrt{n}} \sum_{i=1}^n 
  \begin{pmatrix}
W_i - \mu \\
\bm{Z}_i \\
W_i^2 - \E W_1^2 
\end{pmatrix} \\ 
& =:\bm{B}_n   \frac{1}{\sqrt{n}} \sum_{i=1}^n \bm{C}_i.
\end{split} 
\end{equation}
For the asymptotic normality of \eqref{eq:def:bn:ci} we use the Cram\'er-Wold device and show that
\begin{equation*}
\sqrt{n}\Big( \frac{1}{n} \sum_{i=1}^n \bm{\lambda}^\top \bm{C}_i \Big) \std \mathcal{N}(0,\bm{\lambda}^\top \bm{\Sigma}^{\ast} \bm{\lambda}),\quad\nto,
\end{equation*}
for all vectors $\bm{\lambda} \in \R^{r+2}$ such that $\bm{\lambda}^\top \bm{\Sigma}^{\ast} \bm{\lambda}>0$. 
It follows from Proposition~\ref{th:haug}(c) that the squared returns process $(W_i)_{i \in \N}$ is $\alpha$-mixing with exponentially decaying mixing coefficients.
Since each $\bm{C}_i$ is a measurable function of $W_i,\dots,W_{i-r}$ it follows from Remark~1.8 of \citet{Bradley07} that $(\bm{\lambda}^\top \bm{C}_i)_{i \in \N}$ is also $\alpha$-mixing with mixing coefficients satisfying $\alpha_{\bm{C}}(n) \leq \alpha_{\bm{W}}(n-(r+1))$ for all $n \geq r + 2$. Therefore $\sum_{n=0}^{\infty} (\alpha_{\bm{C}}(n))^{\frac{\epsilon}{2 + \epsilon}} < \infty$ for all $\epsilon > 0$. 
Since $\E |L_1|^{8+\epsilon} < \infty$ and $\Psi_{\bstheta}(4+\frac{\epsilon}2) < 0$ it follows from \eqref{eq:2.5} that $\E |W_1|^{4+\epsilon/2} < \infty$ and, as a consequence, $\E |\lambda^\top \bm{C}_1|^{2+\epsilon/4} < \infty$. 
Thus, the CLT for $\alpha$-mixing sequences applies (see Theorem~18.5.3 of \citet{Ibragimov71}) so that 
\begin{equation*}
\sqrt{n}\Big( \frac{1}{n} \sum_{i=1}^n \bm{\lambda}^\top \bm{C}_i \Big) \std \mathcal{N}(0,\zeta),\quad\nto,
\end{equation*}
where
\begin{equation*}
\zeta = \E \bm{\lambda}^\top \bm{C}_1 \bm{C}_1^\top \bm{\lambda}  + 2\sum_{i=1}^{\infty} \E \bm{\lambda}^\top \bm{C}_1 \bm{C}_{1+i}^\top \bm{\lambda}.
\end{equation*}
After rearranging this equation we find  \eqref{eq:defSigAs}.
Let $\bm{B}_n = (b^n_{u,v})_{u,v=1}^{r+2}$ denote the matrix as defined in \eqref{eq:def:bn:ci}. Using \eqref{eq:wTw_lim} we get, as $\nto$, 
$$
b^n_{u,v} \stas \E (W_{1}-\mu)(W_{1+|u-v|} - \mu), \quad \quad 2 \leq u,v \leq r+1.
$$
Then the inner block of the matrix $\bm{B}_n$ converges a.s. to $\bm{\Gamma}^{-1}$. 
 This gives \eqref{eq:defSig}
which finishes the proof. 
\end{proof}

%
%

\section{IIE of the COGARCH process - strong consistency and asymptotic normality}

The objective of this section is to prove strong consistency and asymptotic normality of the IIE of the COGARCH parameter $\bstheta$. 
Let $(G_i(\bm{\theta}_0))_{i=1}^n$ be the returns originating from a COGARCH log-price process \eqref{eq:def:Pt}.
As auxiliary model we will use an AR$(r)$ model for fixed $r\ge 2$ as in Proposition \ref{pr:cog:aux}, whose parameters are estimated by one of the estimators $\hat{\bspi}_n$ from Definition~\ref{de:ywlse}, which we consider as functions of the COGARCH parameter $\bstheta$.

\subsection{Preliminary results} 

{We begin with an auxiliary result, which is a consequence of Theorem~3.2 of \cite{Kluppelberg04}.}

\begin{lemma}\label{le:sigt:dist}
Assume that $\E |L_1|^{2} < \infty$ and $\Psi_{\bm{\theta_i}}(1) < 0$ for $i=1,\dots,d$. Then for every $t > 0$,
\begin{equation*}
(\sigma_t^2(\bm{\theta}_1),\dots,\sigma_t^2(\bm{\theta}_d)) \eqd (\sigma_0^2(\bm{\theta}_1),\dots,\sigma_0^2(\bm{\theta}_d)) 
\end{equation*}
\end{lemma}


In what follows we shall need for fixed $\vp > 0$ the stochastic process
\begin{equation}\label{eq:def_ds}
K_s(\vp) = \sumzerous \frac{ (\Delta L_u)^2}{1 + \vp(\Delta L_u)^2}, \quad s \geq 0.
\end{equation}


\begin{lemma}
The process $(K_s(\vp))_{s \geq 0}$ is a L\'evy process and $\E |K_s(\vp)|^p < \infty$ for all $p \in \N$.  
\end{lemma}
\begin{proof}
That $(K_s(\vp))_{s \geq 0}$ is a L\'evy process is clear. Since 
\begin{equation}\label{eq:supJump}
\sup_{s \geq 0} |\Delta K_s(\vp)| = \sup_{s \geq 0} \frac{ (\Delta L_s)^2}{1 + \vp(\Delta L_s)^2} \leq \frac{1}{\vp} < \infty
\end{equation}
it follows that $(K_s(\vp))_{s \geq 0}$ has bounded jumps and therefore it has moments of all orders (see e.g. Theorem~2.4.7 of \cite{Applebaum09}). 
\end{proof}

For $p \geq 1 $ consider the sets
\begin{equation}\label{eq:def:Thk}
\Theta^{(p)} \subset \calm^{(p)} := \{ \bstheta \in \calm: \Psi_{\bm{\theta}}\big(p \big) < 0\},
\end{equation}
where $\Theta^{(p)}$ is compact, and recall from \eqref{eq:2.5} that the condition $\Psi_{\bm{\theta}}\big(p) < 0$ implicitly requires $\E |L_1|^{2p } < \infty$.


\begin{lemma}\label{le:yt:ineq}
Let $p \geq 1$ and $t \geq 0$ be fixed and consider the sets $\Theta^{(p)}$ and $\calm^{(p)}$ as in \eqref{eq:def:Thk}. Then
\begin{itemize}
\item[(a)] there exist numbers $\bstheta^{\ast}_1,\dots,\bstheta^{\ast}_N \in \calm^{(p)}$ such that
\begin{equation*}\label{eq:bsupEmYs}
\sup_{\bstheta \in \Theta^{(p)}} e^{-Y_t(\bm{\theta})} \leq  \sum_{j=1}^N e^{-Y_t(\bstheta^{\ast}_j)}.
\end{equation*}
\item[(b)]\emph{(Proposition~2 in \citet{Kluppelberg06})}
there exists some $\si^*>0$ such that $\si_0(\bstheta)\ge \si^*$ a.s. for all $\bstheta\in\Theta$.
\end{itemize}
\end{lemma}

\begin{proof}
(a) We use a Heine-Borel argument to control the exponential term. Since $\Theta^{(p)}$ is compact we can find a finite collection of open sets $(\Theta^{(p)}_j)_{j=1}^N$ such that $\Theta^{(p)} \subseteq \cup_{j=1}^N \Theta^{(p)}_j \subset \calm^{(p)}$. For each fixed $j$ the closure $\ov{\Theta^{(p)}_j}$ is a subset of $\calm^{(p)}$ and therefore there exists a point $\bstheta_j^{\ast} = (\beta_j^{\ast},\eta_j^{\ast},\vp_j^{\ast})^\top \in \calm^{(p)}$ such that $\eta \geq \eta_j^{\ast}, \vp \leq \vp_j^{\ast}$ for all $\bstheta \in \Theta^{(p)}_j$.
This implies that for all $\bstheta \in \ov{\Theta^{(p)}_j}$:
\begin{equation*}\label{eq:YtHB}
Y_t(\bstheta) = \eta t - \sumzerout \log ( 1 + \vp (\Delta L_u)^2 ) \geq \eta_j^{\ast} t - \sumzerout \log ( 1 + \vp_j^{\ast} (\Delta L_u)^2 ) = Y_t(\bstheta_j^{\ast}), \quad t \geq 0.
\end{equation*}
Therefore, 
$$
\sup_{\bstheta \in \Theta^{(p)}} e^{-Y_t(\bm{\theta})}  \leq \sum_{j=1}^N e^{-Y_t(\bstheta^{\ast}_j)},
$$
proving the statement.
\end{proof}

\begin{remark}\label{rem:f}
Both $\hatbspi_{n,\yw}$ (by \eqref{eq:3.6}) and $\hatbspi_{n,\ls}$ (as in the proof of Proposition~5.6 in \citet{Fasen18II}) can be written as a map $g:\R^{r+2} \rightarrow \R^{r+2}$ for $r\ge 2$ {with $g(\bsx)=(g_1(\bsx),\dots,g_{r+2}(\bsx))$ for $\bsx=(x_1,\dots,x_{r+2})$ applied to the vector}
\begin{equation}\label{eq:thIICon1}
f_n(\bstheta) = \Big( \frac{1}{n} \sum_{i=1}^n G^2_i(\bm{\theta}), \frac{1}{n}\sum_{i=1}^{n-h} G^2_i(\bm{\theta})G^2_{i+h}(\bm{\theta}), h = 0,\dots,r \Big), \quad n\in\N.
 \end{equation}
Since $g$ involves only matrix multiplications and matrix inversion of non-singular matrices, {it inherits the smoothness properties of $G_i(\bstheta)$ for $i\in\N$. }
Since $(G_i(\bstheta))_{i \in \N}$ is stationary and ergodic, Birkhoff's ergodic theorem applies and $f_n(\bstheta)$ converges as $n\to\infty$ pointwise to
 \beam\label{eq:f}
 f(\bstheta) = (\E G_1^2(\bstheta), \E G_1^2(\bstheta)G_{1+h}^2(\bstheta),h=0,\dots,r).
 \eeam
\end{remark}

\begin{remark}\label{re:dim2The}
The results that follow are related to continuity and differentiability of the random elements $(G_i(\bstheta),\bstheta\in\bsTheta)$ for $i\in\N$ with respect to $\bstheta$.
According to \eqref{eq:def:sigt} and \eqref{eq:def:sig0} we find
$$
G_i(\bstheta) = \int_{(i-1)\Delta}^{i\Delta} \sigma_s((\beta,\eta,\vp)) \diff L_s = \sqrt{\beta} \int_{(i-1)\Delta}^{i\Delta} \sigma_s((1,\eta,\vp)) \diff L_s = \sqrt{\beta}\,G_i((1,\eta,\vp)),
$$
which is linear in $\sqrt{\beta}$, hence, $(G_i(\bstheta),\bstheta\in\bsTheta)$ is obviously continuous in $\beta$ and has a partial derivative with respect to  $\beta>0$.
\end{remark}

\subsection{Strong consistency of the IIE}

To ensure strong consistency of $\hat{\bm{\theta}}_{n,\ii}$, we need to verify that $\hat{\bspi}_n(\bstheta)$ satisfies the uniform SLLN of Proposition~\ref{th:ii:asymp}(a). 
The results of Lemma~\ref{le:mu:cov} and Theorem~\ref{th:lse:con} guarantee point-wise strong consistency. 
Uniform strong consistency  will hold by continuity of $g$ (cf. Lemma~\ref{le:pi_f}), if we can apply a uniform SLLN to the sequence in \eqref{eq:thIICon1}.
 
Since the sequence of random elements $(G_i(\bstheta), \bstheta \in \Theta)_{i \in \N}$ is stationary and ergodic, we need to show (cf. Theorem~7 in \citet{Straumann06}) that   $G_i(\bstheta)$ is for every $i \in \N$ a continuous function of $\bstheta$ on $\Theta$ or on some compact subset $\Theta^{(p)}$ of $\Theta$ and that
\begin{equation*}
\E  \sup_{\bstheta \in \Theta^{(p)}} G^4_i(\bstheta)  < \infty.
\end{equation*}
Proving that $G_i(\bstheta)$ is $\omega$-wise continuous in its parameter $\bstheta$ is not straightforward, since
\begin{equation*}
G_i(\bstheta) = \int_{(i-1)\Delta}^{i\Delta} \sigma_s(\bm{\theta}) \diff L_s
\end{equation*}
is a stochastic integral, driven by an arbitrary L\'evy process, which also drives the stochastic volatility process.
If $L$ has finite variation, we can use dominated convergence to show continuity, but this is not possible when $L$ has infinite variation sample paths; cf. Remark~\ref{finitevar} below.
However, as we shall show in the next result, applying Kolmogorov's continuity criterion, we can always find a version $(\gcont_i(\bstheta))_{i \in \N}$ of the sequence $(G_i(\bstheta))_{i \in \N}$, which is continuous on a possibly smaller compact parameter space $\Theta^{(p)} \subseteq \Theta$ for $\Theta$.

\begin{theorem}[H\"older continuity]\label{th:cog:cont}
Assume that $\E |L_1|^{2p(1+\epsilon)} < \infty$ for some $p > 2$ and $\epsilon > 0$. Then there exists a version $(\gcont_i(\bstheta))_{i \in \N}$ of the random elements $(G_i(\bstheta))_{i \in \N}$ which is H\"older continuous of every order $\gamma\in [0,(p-2)/(2p))$ on $\ThetaHCone$ as defined in \eqref{eq:def:Thk}.
Additionally, for every $q \in [0,2p)$, $i \in \N$, and for 
\begin{equation}\label{eq:DefUi}
U_i = \supThetaOnePLoc \frac{ |\gcont_i(\bstheta_1) - \gcont_i(\bstheta_2)|}{\|\bstheta_1 - \bstheta_2 \|^{\gamma}}
\end{equation}
we have
\begin{equation}\label{eq:Esu2}
\E  \sup_{\bstheta \in \ThetaHCone} |\gcont_i(\bstheta)|^q  < \infty \quad \text{and} \quad \E  U_i^q  < \infty.
\end{equation}
\end{theorem}
\begin{proof}
Without loss of generality we prove this for $i = 1$. We find a continuous version of the random element $G_1(\bstheta)$ on $\ThetaHCone$. 
We first prove continuity with respect to $(\eta,\vp)$ and assume that $\bstheta_1, \bstheta_2 \in \ThetaHCone$ with $\beta_1,\beta_2 = 1$. Using the simple inequality $|a-b|^{2p} \leq |a^2-b^2|^p$, the stationarity of $\sigma_0(\bstheta)$ in Lemma~\ref{le:sigt:dist}, its differentiability in \eqref{eq:c.1} and the mean value theorem gives
\begin{equation}\label{eq:E:dif:sig}
\begin{split}
\int_0^{\Delta} \E  |  \sigma_s(\bstheta_1) -  \sigma_s(\bstheta_2)|^{2p} \, \diff s & \leq 
\int_0^{\Delta} \E  |  \sigma^2_s(\bstheta_1) -  \sigma^2_s(\bstheta_2)|^p \\
&  = \Delta \E |  \sigma^2_0(\bstheta_1) -  \sigma^2_0(\bstheta_2)|^p  \\
& \leq \Delta \Big( \E \sup_{\bstheta \in \ThetaHCone} | \gradtwo \sigma^2_0(\bstheta) |^p \Big)   \|\bstheta_1 - \bstheta_2\|^{p}  < \infty
\end{split}
\end{equation}
by Lemma~\ref{le:bEsupSnB} with $k = 1$. By (A2) of Proposition~\ref{th:haug}, $(L_t)_{t \geq 0}$ is a martingale. Since $\E |L_1|^{2p} < \infty$ and $\int_0^{\Delta} \E |  \sigma_s(\bstheta_1) -  \sigma_s(\bstheta_2)|^{2p}  \diff s < \infty$ we can apply Theorem~66 of Ch.~5 in \citet{Protter90} to the stochastic integral in \eqref{eq:def:Pt} and obtain
$$
\E |G_1(\bstheta_1) - G_1(\bstheta_2)|^{2p}
=  \E \Big|  \int_0^{\Delta} ( \sigma_s(\bm{\theta}_1) -  \sigma_s(\bm{\theta}_2) ) \diff L_s  \Big|^{2p}
\leq c^{\ast} \int_0^{\Delta} \E |  \sigma_s(\bstheta_1) -  \sigma_s(\bstheta_2)|^{2p} \diff s,
$$
where $c^{\ast}$ is a positive constant. 
This combined with \eqref{eq:E:dif:sig} gives 
\begin{equation}\label{eq:thcogco1}
\E |G_1(\bstheta_1) - G_1(\bstheta_2)|^{2p} \leq c{\Delta}\|\bstheta_1 - \bstheta_2\|^{p},
\end{equation}
where $c = c^\ast\Delta\E \sup_{\bstheta \in \ThetaHCone} | \gradtwo \sigma^2_0(\bstheta) |^p$. Since $\beta_1,\beta_2=1$ we show continuity with respect to $(\eta,\vp)$; i.e. the parameter space has dimension $d=2$. 
Since $p >  2=d$ we can apply Kolmogorov's continuity criterion (Theorem 10.1 in \citet{Schilling14BM}, see also Theorem~2.5.1 of Ch. 5 in \citet{Khoshnevi02}).
Then there exists a version $(G_1^{(c)}(\bstheta),\bstheta\in\ThetaHCone)$ of $(G_1(\bstheta),\bstheta\in \ThetaHCone)$ which is H\"older continuous of every order $\gamma\in [0,(p-2)/(2p))$; hence, also continuous. 
Since $\Theta$ is compact, Lemma~\ref{le:kol_mo} together with \eqref{eq:thcogco1} gives $\E  \sup_{\bstheta \in \ThetaHCone} |\gcont_1(\bstheta)|^q  < \infty$ for every $q \in [0,2p)$. Finally, the second expectation in \eqref{eq:Esu2} is finite by Theorem 10.1 in \cite{Schilling14BM}.

Now because of Remark~\ref{re:dim2The}, $\gcont_1(\bstheta)$ is linear in $\sqrt{\beta}$ and therefore the results can be generalized for the map $\bstheta \mapsto \gcont_1(\bstheta)$ on $\ThetaHCone$. Indeed, let $\beta^\ast$ be as in \eqref{eq:betaUas} and 
\begin{equation*}\label{eq:dfbelowa}
\beta_\ast = \inf \{\beta > 0: (\beta,\eta,\vp) \in \Theta\} > 0.
\end{equation*}
 Now, for arbitrary $\bstheta_1,\bstheta_2 \in \ThetaHCone$ we can use Remark~\ref{re:dim2The}, the mean value theorem for $\beta \mapsto \sqrt{\beta}$ and the H\"older continuity of order $\gamma$ of $(\eta,\vp) \mapsto \gcont_1((1,\eta,\vp))$, the definition of the $\ell^1$-norm and the fact that $\gamma \in (0,1)$ to get
\begin{equation}\label{eq:230}
\begin{split}
\quad\quad\quad & |\gcont_1(\bstheta_1) - \gcont_1(\bstheta_2)| \\
& \leq |\gcont_1((1,\eta_1,\vp_1)) - \gcont_1((1,\eta_2,\vp_2))| \sqrt{\beta^\ast} + \frac{1}{2\sqrt{\beta_{\ast}}}|\beta_1 - \beta_2|\sup_{\bstheta \in \ThetaHCone} |\gcont_1(\bstheta)|  \\
& \leq K \|(\eta_1,\vp_1) - (\eta_2,\vp_2)\|^\gamma \sqrt{\beta^\ast} + \frac{1}{2\sqrt{\beta_{\ast}}}|\beta_1 - \beta_2|^{\gamma} |2\beta^\ast|^{1-\gamma}  \sup_{\bstheta \in \ThetaHCone} |\gcont_1(\bstheta)|  \\
& \leq \|\bstheta_1 - \bstheta_2\|^\gamma \Big( K\sqrt{\beta^\ast} +   \frac{1}{2\sqrt{\beta_{\ast}}} |2\beta^\ast|^{1-\gamma} \sup_{\bstheta \in \ThetaHCone} |\gcont_1(\bstheta)|  \Big),
\end{split}
\end{equation}
showing the H\"older continuity of $\bstheta \mapsto \gcont_1(\bstheta)$ on $\ThetaHCone$. Now, the first expectation in \eqref{eq:Esu2} is finite since $|\gcont_1(\bstheta)| \leq \beta^\ast |\gcont_1((1,\eta,\vp))|$. Now let $\bstheta_1,\bstheta_2 \in \ThetaHCone$ be such that $0 < \|\bstheta_1-\bstheta_2\| < 1$.  Using the inequality at the first line of \eqref{eq:230} and the definition of the $\ell^1$-norm gives
\begin{equation}\label{eq:222}
\begin{split}
\quad\quad & \supThetaOnePLoc \frac{|\gcont_1(\bstheta_1) - \gcont_1(\bstheta_2)|}{\|\bstheta_1 - \bstheta_2\|^\gamma} \\
 & \leq \Bigg( \supThetaOnePLoc \frac{|\gcont_1((1,\eta_1,\vp_1) - \gcont_1((1,\eta_2,\vp_2)|}{\|\bstheta_1 - \bstheta_2\|^\gamma}  \Bigg) \sqrt{\beta^\ast}\\
 & \quad\quad\quad\quad\quad\quad\quad\quad + \supThetaOnePLoc \frac{|\beta_1 - \beta_2|}{2\sqrt{\beta_{\ast}}\|\bstheta_1-\bstheta_2\|^\gamma} |\gcont_1(\bstheta)| \\
 & \leq \Bigg( \supThetaOnePLoc \frac{|\gcont_1((1,\eta_1,\vp_1) - \gcont_1((1,\eta_2,\vp_2)|}{\|(\eta_1,\vp_1) - (\eta_2,\vp_2)\|^\gamma}  \Bigg) \sqrt{\beta^\ast} + \frac{1}{2\sqrt{\beta_{\ast}}}\sup_{\bstheta \in \ThetaHCone} |\gcont_1(\bstheta)| \\ 
\end{split}
\end{equation}
Applying the supremum and raising both sides of \eqref{eq:222} to the power $q$ gives the result.
\end{proof}

\begin{remark}\label{re:cont:G}
In view of Theorem~\ref{th:cog:cont} we will from now on work with a  continuous version of the  returns $(G_i(\bstheta),\bstheta\in\ThetaHCone)_{i \in \N}$.
\end{remark}

\begin{theorem}[Strong consistency of the IIE]\label{th:ii:cons}
Assume that $\E |L_1|^{2p(1+\epsilon)} < \infty$ for some $p > 2$ and $\epsilon > 0$ and let $(G_i(\bstheta_0))_{i=1}^n$ be the  returns \eqref{eq:def:Gi} with parameter $\bstheta_0 \in \ThetaHCone$ from \eqref{eq:def_ds}. 
Suppose that the auxiliary AR$(r)$ model for $r\ge 2$ is estimated by the LSE or the YWE of Definition \ref{de:ywlse}. Then 
\begin{equation*}
\hat{\bm{\theta}}_{n,\ii} \overset{\text{a.s.}}{\rightarrow} \bm{\theta}_0, \quad \nto.
\end{equation*}
\end{theorem}
\begin{proof}
According to Proposition~\ref{th:ii:asymp}(a), strong consistency of the IIE will follow if as $\nto$,
$$
\sup_{\bm{\theta} \in \ThetaHCone}  \| \hatbspi_{n,\ls}(\bm{\theta}) - \bspi_{\bstheta}  \| \overset{\text{a.s.}}{\rightarrow} 0 \quad \text{and} \quad
\sup_{\bm{\theta} \in \ThetaHCone}  \| \hatbspi_{n,\yw}(\bm{\theta}) - \bspi_{\bstheta}  \| \overset{\text{a.s.}}{\rightarrow} 0.
$$
By Remark~\ref{rem:f} and  Lemma~\ref{le:pi_f} it suffices to prove that 
\beam\label{supf}
\sup_{\bstheta \in \ThetaHCone} \|f_n(\bstheta) - f(\bstheta) \| \stas 0,\quad \nto,
\eeam
for $f_n$ and $f$ as defined in \eqref{eq:thIICon1} and \eqref{eq:f}, respectively.
The Cauchy-Schwarz inequality gives for every $h \in \N_0$, 
\beam\label{eq:supgigh}
\E \sup_{\bm{\theta} \in \ThetaHCone} G^2_1(\bm{\theta})G^2_{1+h}(\bm{\theta})  
& \le & \Big( \E  \sup_{\bm{\theta} \in \ThetaHCone} G^4_1(\bm{\theta}) \Big)^{\frac{1}{2}}
\Big( \E  \sup_{\bm{\theta} \in \ThetaHCone} G^4_{1+h}(\bm{\theta}) \Big)^{\frac{1}{2}} < \infty.
\eeam
The right-hand side of \eqref{eq:supgigh} is finite by Theorem~\ref{th:cog:cont}. It also follows from the same theorem that $\E  \sup_{\bm{\theta} \in \ThetaHCone} G^2_1(\bm{\theta}) < \infty$ and, hence, by Theorem~7 in \citet{Straumann06} the uniform SLLN holds and we obtain for all $h \in \N_0$ as $\nto$,
\begin{equation}\label{eq:thIICon2}
\begin{split}
& \sup_{\bm{\theta} \in \ThetaHCone} \Big | \frac{1}{n} \sum_{i=1}^n G^2_i(\bm{\theta}) - \E  G^2_1(\bm{\theta})  \Big |  \overset{\text{a.s.}}{\rightarrow}  0 \quad \text{and} \\
& \sup_{\bm{\theta} \in \ThetaHCone} \Big | \frac{1}{n} \sum_{i=1}^n G^2_i(\bm{\theta})G^2_{i+h}(\bm{\theta}) - \E  G^2_1(\bm{\theta})G^2_{1+h}(\bm{\theta})  \Big | \overset{\text{a.s.}}{\rightarrow} 0.
\end{split}
 \end{equation}
Hence \eqref{supf} follows from \eqref{eq:thIICon2} 
finishing the proof.
\end{proof}

\begin{remark}\label{finitevar}
If the L\'evy process $(L_t)_{t\geq 0 }$ has finite variation sample paths, then the stochastic integral in \eqref{eq:def:Gi} can be treated pathwise as a Riemann-Stieltjes integral, such that continuity of $(G_i(\bstheta),\bstheta\in \Theta)_{i \in \N}$ follows from Lemma~\ref{le:yt:ineq}(c) and dominated convergence. Therefore, Theorem~\ref{th:cog:cont} is valid for $\bstheta_0 \in \Theta \supseteq$ $\ThetaHCone$ for $p > 2$ and some $\epsilon > 0$.
Additionally, since the total variation process is also a L\'evy process we can use Theorem~66 of Ch.~5 in \citet{Protter90} to show that $\E L^4_1 < \infty$ implies $\E  \sup_{\bstheta \in \Theta}G^4_i(\bstheta) < \infty$ for all $i \in \N$. Therefore, also Theorem~\ref{th:ii:cons} is valid for $\bstheta_0 \in \Theta \supseteq \ThetaHCone$.
\end{remark}

%
%
\subsection{Asymptotic normality of the IIE}

In order to prove asymptotic normality of the IIE, we need to verify the conditions (b.1), (b.2) and (b.3) of Proposition~\ref{th:ii:asymp}. 
We recall that (b.2) has been proved in Theorem~\ref{th:ar:norm}, and it remains to prove (b.1) and (b.3), which are related to the smoothness of $\hatbspi_n(\bstheta)$ as a function of $\bstheta$.

\subsubsection{Differentiability properties of $(G_i(\bstheta),\bstheta\in\ThetaHCone)$}

Condition (b.1),  refers to the differentiability of the map $\hatbspi_n(\bstheta)$.
By Remark~\ref{rem:f} and the chain rule we only need to prove differentiability of $G_i(\bstheta)$ with respect to $\bstheta$ for every $i \in \N$. 
Since $G_i(\bstheta)$ is defined in terms of a stochastic integral we can not simply interchange the order of  the Riemann differentiation and the stochastic integration, however,  under appropriate regularity conditions formulated in \citet{James84} this is possible. 

We start by investigating the candidate for the differential of $(G_i(\bstheta),\bstheta\in\ThetaHCone)$ with $\ThetaHCone$ as in \eqref{eq:def:Thk}, namely the map

\begin{equation}\label{eq:def:grdG}
\bstheta \mapsto  \int_{(i-1)\Delta}^{i\Delta} \gradtheta \sigma_s(\bstheta) \diff L_s := \gradtheta G_i(\bstheta).
\end{equation}
We show in Lemma~\ref{le:cog:diff} that we can find a version of the integral on the rhs, which is continuous on a subset $\ThetaHCtwo$ of $\ThetaHCone$. 
Then, Theorem~\ref{le:cog:dif2} asserts that $G_i(\bstheta)$ is differentiable on $\ThetaHCtwo$ and that its differential is indeed given by \eqref{eq:def:grdG}.

\begin{lemma}[H\"older continuity of derivatives]\label{le:cog:diff}
Assume that $\E |L_1|^{4p(1+\epsilon)} < \infty$ for some $p > 2$ and $\epsilon > 0$. Then there exists a version $(\gradtheta \gcont_i(\bstheta))_{i \in \N}$ of the random elements $(\gradtheta G_i(\bstheta))_{i \in \N}$ which is H\"older continuous of every order $\gamma\in [0,(p-2)/p)$ on $\ThetaHCtwo$ as defined in \eqref{eq:def:Thk}.
Additionally, for every $q \in [0,p)$, $i \in \N$, $l \in \{1,2,3\}$, and for 
\begin{equation}\label{eq:DefVi}
V_i = \supThetaTwoPLoc \frac{ |\partiall G_i(\bstheta_1) - \partiall G_i(\bstheta_2)|}{\|\bstheta_1 - \bstheta_2 \|^{\gamma}}
\end{equation}
we have
\begin{equation*}\label{eq:EsuVi}
\E  \sup_{\bstheta \in \ThetaHCtwo} \Big|\partiall G_i(\bstheta)\Big|^q  < \infty \quad \text{and} \quad \E  V_i^q  < \infty.
\end{equation*}
\end{lemma}
\begin{proof}
Without loss of generality we consider $i = 1$. Note that in view of Remark~\ref{re:dim2The} we can write \eqref{eq:def:grdG} as
\begin{equation}\label{eq:def:grdG2}
 \Big(   \frac{1}{2\sqrt{\beta}}  \int_0^{\Delta}\sigma_s((1,\eta,\vp))  \diff L_s, \sqrt{\beta} \int_0^{\Delta} \frac{\partial}{\partial \eta} \sigma_s((1,\eta,\vp)) \diff L_s, 
\sqrt{\beta} \int_0^{\Delta} \frac{\partial}{\partial \vp} \sigma_s((1,\eta,\vp))  \diff L_s \Big)^\top.
\end{equation}
From Remark~\ref{re:cont:G} the first component of \eqref{eq:def:grdG2} is continuous in $\beta$ even on $\Theta$.
 For the remaining two components, we show continuity with respect to $(\eta,\phi)$. Thus, assume that $\bstheta_1, \bstheta_2 \in \ThetaHCtwo$ with $\beta_1,\beta_2 = 1$.

Using the distributional property of $\sigma_s(\cdot)$ in Lemma~\ref{le:sigt:dist} and the differentiability of $\bstheta \mapsto \sigma_s(\bstheta)$ in Lemma~\ref{le:sigt:diff} gives for every Borel set $B \in \mathcal{B}(\R)$ that
\begin{equation*}\label{eq:diff_sig}
\begin{split}
& \quad\,\, \P \Big( \frac{\partial}{\partial \eta}  ( \sigma_s(\bstheta_1) - \sigma_s(\bstheta_2) ) \in B \Big)\\
& = \P \Big( \lim_{\htzero, h \in \Q} \frac{ [ \sigma_s(\bstheta_1 + (h,0))- \sigma_s(\bstheta_1) ] - [ \sigma_s(\bstheta_2 + (h,0))- \sigma_s(\bstheta_2)]}{h}  \in B \Big)\\
& = \lim_{\htzero, h \in \Q} \P \Big(  \frac{ [ \sigma_s(\bstheta_1 + (h,0))- \sigma_s(\bstheta_1) ] - [ \sigma_s(\bstheta_2 + (h,0))- \sigma_s(\bstheta_2)]}{h}  \in B \Big)\\
& = \lim_{\htzero, h \in \Q} \P \Big(  \frac{ [ \sigma_0(\bstheta_1 + (h,0))- \sigma_0(\bstheta_1) ] - [ \sigma_0(\bstheta_2 + (h,0))- \sigma_0(\bstheta_2)]}{h}  \in B \Big)\\
& = \P \Big( \frac{\partial}{\partial \eta}  ( \sigma_0(\bstheta_1) - \sigma_0(\bstheta_2) ) \in B \Big),
\end{split}
\end{equation*}
so that 
\begin{equation}\label{eq:lecd1}
\frac{\partial}{\partial \eta}  ( \sigma_s(\bstheta_1) - \sigma_s(\bstheta_2)) \eqd  \frac{\partial}{\partial \eta}  ( \sigma_0(\bstheta_1) - \sigma_0(\bstheta_2)).
\end{equation}
Similar calculations show that \eqref{eq:lecd1} is also valid for $\frac{\partial}{\partial \eta}$ replaced by $\frac{\partial}{\partial \vp}$. Thus, $(\gradtwo \sigma_s(\bstheta))_{s \geq 0}$ is stationary and it follows from its differentiability in \eqref{eq:c.2} and the mean value theorem that
\begin{equation*}\label{eq:E:dif:sig2}
\begin{split}
\int_0^{\Delta} \E  \|  \gradtwo \sigma_s(\bstheta_1) -  \gradtwo \sigma_s(\bstheta_2)\|^{p} \, \diff s &  = \Delta \E \|  \gradtwo \sigma_0(\bstheta_1) -  \gradtwo \sigma_0(\bstheta_2)\|^p  \\
& \leq \Delta \Big( \E \sup_{\bstheta \in \ThetaHCone} \| \gradtwo^2 \sigma_0(\bstheta) \|^p \Big)   \|\bstheta_1 - \bstheta_2\|^{p}  < \infty,
\end{split}
\end{equation*}
by Lemma~\ref{le:bEsupSnB} with $k = 2$. The rest of the proof follows along the same lines as in Theorem~\ref{th:cog:cont}.
\end{proof}

%
%

%
%

\begin{theorem}[Differentiable version of $(G_i(\bstheta))_{i\in\N}$]\label{le:cog:dif2}
Assume the conditions of Lemma~\ref{le:cog:diff}. 
Then there is a version $(G_i(\bstheta),\bstheta\in\ThetaHCtwo)_{i\in\N}$ for $\ThetaHCtwo$ as in \eqref{eq:def:Thk}, which is continuously differentiable and its derivative is given a.s. by $(\gradtheta G_i(\bstheta),\bstheta\in\ThetaHCtwo)_{i\in\N}$.
\end{theorem}

\begin{proof}
Without loss of generality we consider $i = 1$. From Remark~\ref{re:dim2The} it follows that $G_1(\bstheta) = \sqrt{\beta} G_1((1,\eta,\vp))$ so that obviously
$$
\frac{\partial}{\partial \beta} G_1(\bstheta) = \frac{1}{2\sqrt{\beta}} G_1((1,\eta,\vp)) = \int_0^{\Delta} \frac{\partial}{\partial \beta} \sigma_s(\bstheta) \diff L_s.
$$
Interchanging the partial differentiation with respect to $(\eta,\phi)$ and the stochastic integral requires the four regularity conditions of Theorem~2.2 in \citet{James84}. 
Let  $\calf_t:=\sigma( \{ L_s, 0 \leq s \leq t\})$, such that $(\calf_t)_{t \ge 0}$ is the filtration generated by the L\'evy process $L$.
Condition (i) of that paper is satisfied, since $(\sigma_t(\bstheta))_{t\ge0}$ is predictable, we consider the parameter space $\calm$ with the Borel $\sigma$-algebra, and the parameter $\bstheta$ is independent of $t$. 
Since $\sigma_s(\bstheta) = \sqrt{\beta} \sigma_s((1,\eta,\vp))$ these regularity conditions need only to be checked for the map $(\eta,\vp) \mapsto \sigma_s((1,\eta,\vp))$. 
Condition (ii) requires that $\int_0^{\Delta} \sigma^2_s(\bstheta) \diff \langle L \rangle_s < \infty$ a.s. for every $\bstheta\in\ThetaHCtwo$, where $\langle L \rangle=(\langle L\rangle_s)_{s\ge 0}$ is the characteristic of the martingale $L$. Since $L$ is a square integrable L\'{e}vy process, $ \langle L \rangle_s = s \E L_1^2$ and thus this condition holds since $s \mapsto \sigma_s(\bstheta)$ has bounded sample paths on the compact $[0,\Delta]$. 
The first part of condition (iii)  requires that for every fixed $s$, the map $\bstheta \mapsto \sigma_s(\bstheta)$ is absolutely continuous.
From the definition of $\sigma_s^2(\bstheta)$ in \eqref{eq:def:sigt} we have for $\beta_1=\beta_2=1$, 
\begin{equation}\label{eq:leCogd1}
\sigma_s^2(\bstheta) = e^{-Y_{s-}(\bstheta)}\Big( \int_0^s e^{Y_v(\bstheta)} \diff v +\int_0^{\infty} e^{-Y_v(\bstheta)} \diff v \Big) =: h(\bstheta)(f(\bstheta) + g(\bstheta) ).
\end{equation}
Then for fixed $\bstheta_1, \bstheta_2 \in \ThetaHCtwo$ we use Lemma~\ref{le:sigt:diff} in combination with the mean value theorem and Lemma~\ref{le:yt:ineq}(a) to get
\beam\label{eq:leCogd2}
& \quad & |\sigma^2_s(\bm{\theta}_1) - \sigma^2_s(\bm{\theta}_2)| \\
& \leq &\Big| (f(\bstheta_1)+ g(\bstheta_1))\Big(h(\bstheta_1)-h(\bstheta_2) + h(\bstheta_2)\big(f(\bstheta_1)-f(\bstheta_2) + g(\bstheta_1)-g(\bstheta_2\big)  \Big)\Big| \notag\\
& \leq & (h(\bstheta_2)+f(\bstheta_1) + g(\bstheta_1))\Big( |h(\bstheta_1)-h(\bstheta_2)| + |f(\bstheta_1)-f(\bstheta_2)| + |g(\bstheta_1)-g(\bstheta_2)|  \Big) \notag\\
& \leq & \sup_{\bstheta \in \Theta}  \{h(\bstheta)+f(\bstheta) + g(\bstheta) \} \Big( |h(\bstheta_1)-h(\bstheta_2)| + |f(\bstheta_1)-f(\bstheta_2)| + |g(\bstheta_1)-f(\bstheta_2)| \Big) \notag\\
& \leq  & \| \bstheta_1 - \bstheta_2\| \sum_{j=1}^N \Big\{ \sup_{\bstheta \in \Theta} (h(\bstheta)+f(\bstheta) + g(\bstheta)) \Big\} 
\Big\{e^{-Y_s(\bstheta_j^{\ast})}(s+K_s(\vp_\ast)) \notag\\
 & + &  \int_0^s e^{Y_v(\bstheta_j^{\ast})}(v+K_v(\vp_\ast)) \diff v  
   + \int_0^{\infty}   e^{-Y_v(\bstheta_j^{\ast})}(v+K_v(\vp_\ast)) \diff v   \Big\},\notag
\eeam
where $(\bstheta_j^\ast)_{j=1}^N \in \calm^{(2p(1+\epsilon))}$.
Since $\Theta$ is compact and for each fixed $s \geq 0$, $\bstheta \mapsto \sigma_s(\bstheta)$ is continuous, $ \sup_{\bstheta \in \Theta} \{h(\bstheta)+f(\bstheta) + g(\bstheta)\}$ is finite. 
Furthermore, Lemma~\ref{le:ws:mom2} implies that the other three random variables at the right-hand side of \eqref{eq:leCogd2} have finite first moment, and are therefore also a.s. finite. Thus \eqref{eq:leCogd2} implies that the map $\bstheta \mapsto \sigma^2_s(\bstheta)$ is a.s. Lipschitz continuous on $\ThetaHCtwo$ and, as a consequence, absolutely continuous on $\ThetaHCtwo$.
For the second part of condition (iii) we recall first that we have assumed that $\beta=1$, such that we focus on the partial differentiation of the parameter $(\eta,\vp)^\top$.
A non-decreasing predictable process $(\lambda_t)_{t \geq 0}$ is needed such that for every $t$ and $\bstheta \in \ThetaHCtwo$
\begin{equation*}
\int_0^t \|\gradtwo \sigma_s(\bstheta)\|^2 \diff \langle L \rangle_s < \lambda_t, \quad \text{a.s.}
\end{equation*}
From \eqref{eq:leCogd1}, the product rule {and Proposition~2 in \cite{Kluppelberg06}} we find
\begin{equation}\label{eq:leCogd3}
\| \gradtwo \sigma_s(\bstheta) \| \leq \frac{1}{2\sigma^{\ast}}\{\| \gradtwo h(\bstheta)  \| ( f(\bstheta) + g(\bstheta) ) + h(\bstheta)  \| \gradtwo f(\bstheta) +\gradtwo  g(\bstheta)  \| \}.
\end{equation}
We use Lemma~\ref{le:sigt:diff} and the definition of the process $(Y_t(\bstheta))_{t \geq 0}$ in \eqref{eq:def:yt}.
First note that 
\begin{equation}\label{eq:defEVPStar}
\eta\le\sup\{\eta > 0: (\beta,\eta,\vp) \in \Theta\} =:\eta^\ast < \infty, \quad \vp\geq \inf\{\vp > 0: (\beta,\eta,\vp) \in \Theta\} =:\vp_\ast > 0
\end{equation}
and we get the bound
\begin{equation}\label{eq:leCogd5}
f(\bstheta) = \int_0^s e^{Y_v(\bstheta)} \diff v =  \int_0^s \exp \Big\{ \eta v - \sumzerous \log{ (1 + \vp (\Delta L_u)^2) } \Big\} \diff v \leq s e^{ \eta^\ast s}.
\end{equation}
Hence it follows from \eqref{eq:leCogd1} that
\beam
\label{eq:leCogd6}
\| \gradtwo f(\bstheta)  \| & =&  \int_0^s v e^{Y_v(\bstheta)} \diff v +  \int_0^s e^{Y_v(\bstheta)} K_v(\vp)  \diff v  \leq \Big( s + K_s(\vp_\ast) \Big) s e^{\eta^\ast s},\\
\label{eq:leCogd4}
\| \gradtwo h(\bstheta)  \| & =&  s e^{-Y_{s-}(\bstheta)} +
  e^{-Y_{s-}(\bstheta)} K_s(\vp)  \leq e^{-Y_{s}(\bstheta)}\big( s + K_s(\vp_\ast) \big),\\
\label{eq:leCogd7}
 \| \gradtwo g(\bstheta)  \| & =& \int_0^{\infty} v e^{-Y_v(\bstheta)} \diff v +  \int_0^{\infty} e^{-Y_v(\bstheta)} K_v(\vp)  \diff v \leq \int_0^{\infty}  \big( v + K_v(\vp_\ast) )  e^{-Y_v(\bstheta)} \diff v.
 \eeam

From \eqref{eq:leCogd3} and the bounds given in \eqref{eq:leCogd5}, \eqref{eq:leCogd6},  \eqref{eq:leCogd4}, and \eqref{eq:leCogd7} we obtain
\beam
\label{eq:Hutt4}
 \|\gradtwo \sigma_s(\bstheta)\| &  \leq & \frac{1}{2\sigma^{\ast}} e^{-Y_s(\bstheta)} (s + K_s(\vp_\ast))
\Big( se^{\eta^\ast s} + \int_0^{\infty} e^{-Y_v(\bstheta)} \diff v \Big) \\ 
\nonumber
& + &  \frac{1}{2\sigma^{\ast}} e^{-Ys-(\bstheta)} \bigg\{ ( s +  K_s(\vp_\ast) )s e^{\eta^\ast s}  + \int_0^{\infty} e^{-Y_v(\bstheta)} (v + K_v(\vp_\ast)) \diff v \bigg\} =: l_s(\bstheta).
\eeam
Using the compactness of $\ThetaHCtwo$, \eqref{eq:Hutt4} and Lemma~\ref{le:yt:ineq}(a) gives
$$
\sup_{(\eta,\vp) \in \ThetaHCtwo} \|\gradtwo \sigma_s(\bstheta)\| \leq \sup_{(\eta,\vp) \in \ThetaHCtwo} l_s(\bstheta) \leq \sum_{j=1}^N l_s(\bstheta_j^{\ast}),
$$
where $(\bstheta_j^{\ast})_{j=1}^N$ in $\calm^{(2p(1+\epsilon))}$. 
Thus,
\begin{equation*}
\int_0^t \|\gradtwo \sigma_s(\bstheta)\|^2 \diff \langle L \rangle_s < 1 + \E L^2_1  \int_0^t \big| \sum_{j=1}^N l_s(\bstheta_j^{\ast}) \big|^2 \diff s := \lambda_t,\quad 0\le t\le \Delta,
\end{equation*}
which is a well defined process. 
Since $(\lambda_t)_{t \geq 0}$ is adapted to the filtration $(\calf_t)_{t\ge0}$ and continuous, it is predictable. 
The fourth regularity condition we need to check is that the maps
\begin{equation*}
\bstheta \mapsto  \int_0^{\Delta} \sigma_s(\bm{\theta}) \diff L_s \quad \text{and} \quad \bstheta \mapsto  \int_0^{\Delta} \gradtwo \sigma_s(\bstheta) \diff L_s
\end{equation*}
are continuous, which has been proved in Theorem~\ref{th:cog:cont} and Lemma~\ref{le:cog:diff}. This 
concludes the proof.
\end{proof}

\begin{remark}\label{re:diff:G}
In view of Theorem~\ref{le:cog:dif2} we will from now on work with returns  $( G_i(\bstheta),\bstheta\in\ThetaHCtwo)_{i \in \N}$ with $\ThetaHCtwo$ as in \eqref{eq:def:Thk}, which are continuously differentiable with
\begin{equation*}
\gradtheta G_i(\bstheta) =  \int_{(i-1)\Delta}^{i\Delta}  \gradtheta \sigma_s(\bstheta) \diff L_s, \quad i \in \N.
\end{equation*}
As a consequence, also the map $\bstheta \mapsto \hat{\bspi}_n(\bstheta)$ is continuously differentiable on $\ThetaHCtwo$, hence, condition (b.1) of  Proposition~\ref{th:ii:asymp} holds.
\end{remark}

\subsubsection{Convergence of the derivatives}

%
%

Finally, we prove condition (b.3) of Proposition~\ref{th:ii:asymp}

\begin{proposition}[Consistency of the derivatives]\label{prop:b3}
Assume that $\E |L_1|^{4p(1+\epsilon)} < \infty$ for some $p > 2/5$ and $\epsilon > 0$. Let $\hat{\bspi}_n$ be one of the estimators $\hatbspi_{n,\ls}$ and $\hatbspi_{n,\yw}$ from Definition~\ref{de:ywlse}.
 Then for every sequence $(\bstheta_n)_{n \in \N}\subset \ThetaHCtwo$  and $\bstheta_n \stas \bstheta_0$ we have $\nabla_{\bstheta} \hat{\bspi}_n(\bstheta_n) \stp  \nabla_{\bstheta} \bspi_{\bstheta_0}$ as $\nto$. 
\end{proposition}

\begin{proof}
{Recall from Remark~\ref{rem:f} that we can write the two estimators $\hatbspi_{n,\ls}$ and $\hatbspi_{n,\yw}$  as a continuously differentiable map $g:\R^{r+2} \rightarrow \R^{r+2}$, whose Jacobi matrix exists and all partial derivatives of $g$ are continuous. 
Hence, $\hat{\bspi}_n(\bstheta) = (g_1(f_n(\bstheta)),\dots,g_{r+2}(f_n(\bstheta)))^\top$ for $\bstheta=(\beta,\eta,\varphi)=:(\theta_1,\theta_2,\theta_3)$, and we obtain for the partial derivatives by the chain rule
\begin{equation}\label{eq:leb31}
\frac{\partial}{\partial \theta_l} g_k (f_n(\bstheta)) = \Big( \frac{\partial g_k (f_n(\bstheta))}{\partial x_1},\dots,
\frac{\partial g_k (f_n(\bstheta))}{\partial x_{r+2}} \Big) \Big(\frac{\partial}{\partial \theta_l} f_n(\bstheta)\Big).
\end{equation}
 for every $l = 1,2,3$ and $k = 1,\dots,r+2$. 
By the continuous mapping theorem and \eqref{eq:leb31} it suffices to prove that as $\nto$,}
\begin{equation*}\label{eq:leb32}
f_n(\bstheta_n) \stp f(\bstheta_0) \quad \text{and} \quad \frac{\partial}{\partial \theta_j} f_n(\bstheta_n) \stp \frac{\partial}{\partial \theta_l} f(\bstheta_0), \quad l = 1,2,3.
\end{equation*}
Let  $l \in \{1,2,3\}$ be fixed. It follows from \eqref{eq:thIICon2} and from Lemma~\ref{le:prop55VF1} that  as $\nto$,
\begin{equation}\label{eq:lb322}
\sup_{\bstheta \in \ThetaHCtwo} \|f_n(\bstheta) -  f(\bstheta)\| \stp 0 \quad \text{and} \quad \sup_{\bstheta \in \ThetaHCtwo} \Big\|\partiall f_n(\bstheta) -  \partiall  f(\bstheta)\Big\| \stp 0.
\end{equation}
Since
\begin{equation}\label{eq:fbn}
\begin{split}
\|f_n(\bstheta_n) - f(\bstheta_0)\| & \leq \|f_n(\bstheta_n) -  f(\bstheta_n)\| +  \|f(\bstheta_n) -  f(\bstheta_0)\| \\
 & \leq \sup_{\ThetaHCtwo} \|f_n(\bstheta) -  f(\bstheta)\| +  \|f(\bstheta_n) -  f(\bstheta_0)\|,
\end{split}
\end{equation}
from the continuity of $f$ on $\ThetaHCtwo$, the fact that $\bstheta_n \stp \bstheta_0$, and \eqref{eq:lb322} it follows that $f_n(\bstheta_n) \stp f(\bstheta_0)$. Similar calculations as in \eqref{eq:fbn} show that 
\begin{equation*}
\frac{\partial}{\partial \theta_l} f_n(\bstheta_n) \stp \frac{\partial}{\partial \theta_l} f(\bstheta_0)
\end{equation*}
concluding the proof.

\end{proof}
 
We are now ready to state asymptotic normality of the IIE.

\begin{theorem}[Asymptotic normality of the IIE]\label{th:ii:norm}
Assume that $\E |L_1|^{4p(1+\epsilon)} < \infty$ for some $p > 2/5$ and $\epsilon > 0$. Let $(G_i(\bstheta_0))_{i=1}^n$ be the  returns \eqref{eq:def:Gi} with true parameter $\bstheta_0$ is an element of the interior of $\ThetaHCtwo$ as defined in \eqref{eq:def:Thk}. 
Suppose that the auxiliary AR$(r)$ model for $r\ge 2$ is estimated by the LSE or the YWE of Definition \ref{de:ywlse}. 
If the matrix $\bm{\Sigma} =\bm{\Sigma}_{\bstheta_0}$ defined in Theorem~\ref{th:ar:norm} is positive definite and $\gradtheta \bspi(\bstheta_0)$ has full column rank 3, then 
\begin{equation*}
\sqrt{n}(\hat{\bm{\theta}}_{n,\ii} - \bstheta_0) \std \mathcal{N}(0,\Xi_{\bstheta_0}), \quad \nto,
\end{equation*}
where $\Xi_{\bstheta_0}$ is defined in \eqref{eq:2.19}.
\end{theorem}

\begin{proof}
The asymptotic normality follows from  Proposition~\ref{th:ii:asymp}. 
Since $\bstheta \in \ThetaHCtwo \subseteq  \ThetaHCone\subseteq\Theta$, Theorem~\ref{th:ii:cons} implies condition (a). 
Conditions (b.1) and (b.3) are valid by Proposition~\ref{prop:b3} and the fact that $\gradtheta \bspi(\bstheta_0)$ has full column rank 3. Furthermore, (b.2) holds by Theorem~\ref{th:ar:norm}, since $\bm{\Sigma}_{\bstheta_0}$ is positive definite.
\end{proof}

\begin{remark}[Estimation of the asymptotic covariance matrix of $\hat{\bm{\theta}}_{n,\ii}$]
The asymptotic covariance matrix of $\hat{\bm{\theta}}_{n,\ii}$ given in Theorem~\ref{th:ii:norm} depends on $K,\gradtheta \bspi_{\bstheta_0}, \bm{\Sigma}_{\bstheta_0}$ and $\Omega$. Using the map $\bstheta \mapsto \gradtheta \bspi_{\bstheta}$ from \eqref{eq:def:pi:th} we compute $\gradtheta \bspi_{\hat{\bm{\theta}}_{n,\ii}}$. An application of the continuous mapping Theorem in combination with the continuity of $\bstheta \mapsto \gradtheta \bspi_{\bstheta}$ and Theorem~\ref{th:ii:cons} gives $\gradtheta \bspi_{\hat{\bm{\theta}}_{n,\ii}} \stas \gradtheta \bspi_{\bstheta_0}$. Recall $\bm{\Sigma}_{\bstheta_0}$ as in \eqref{eq:defSig} which depends on the inverse of the autocovariance function $\Gamma_{\bstheta_0}$ and on $\bm{\Sigma}^\ast_{\bstheta_0}$ as defined in \eqref{eq:defSigAs}. A strongly consistent estimator of $\Gamma_{\bstheta_0}$ is given by $\Gamma_{\hat{\bm{\theta}}_{n,\ii}}$. Let $\hat{C_k}$ be as in \eqref{cmatrix} with $W_k$ replaced by $G^2_k(\bm{\theta}_0)$ and $\bspi = (\mu,\bm{a},\gamma(0))$ replaced by $\hatbspi_n$. Then we estimate $\bm{\Sigma}^\ast_{\bstheta_0}$ by
\begin{equation*}
\hat{\mu}_{n,C_1,C_1^T} + 2 \sum_{i=1}^{n-r-1} \hat{\mu}_{n,C_1,C_{1+i}^T},
\end{equation*}
where
\begin{equation*}
\hat{\mu}_{n,C_1,C_{1+i}^T} = \frac{1}{n-i-r}\sum_{k=1}^{n-i-r} C_k C_{k+i}^T, \quad i = 0,\dots,n-r-1.
\end{equation*}
\end{remark}

\begin{remark}
If the L\'evy process $(L_t)_{t\geq 0 }$ has finite variation sample paths, then the stochastic integral in \eqref{eq:def:Gi} can be treated pathwise as a Riemann-Stieltjes integral, such that continuous differentiability of $(G_i(\bstheta),\bstheta\in \Theta)_{i \in \N}$ follows by dominated convergence with dominating function as in \eqref{eq:Hutt4}. Therefore, Theorem~\ref{le:cog:dif2} is valid for $\bstheta_0 \in \Theta \supseteq \ThetaHCtwo$ for some $p > 2$ and $\epsilon > 0$.
Additionally, since the total variation process is also a L\'evy process we can use Theorem~66 of Ch.~5 in \citet{Protter90} to show that, if {$\E L^{8 + \delta}_1 < \infty$ for some $\delta > 0$, then $\E  \sup_{\bstheta \in \ThetaHCone} \| \gradtheta G_i(\bstheta)\|^{4+\delta/2} < \infty$ for all $i \in \N$.} This combined with Remark~\ref{finitevar} and a dominated convergence argument can be applied to show that Lemma~\ref{prop:b3} is valid for $\bstheta_0 \in \ThetaHCtwo$, and, as a consequence, also Theorem~\ref{th:ii:norm}.
\end{remark}

%
%

\section{Simulation study}

The data used for estimation is a sample of COGARCH squared returns $\bm{G}^2_n = (G^2_i(\bstheta_0))_{i=1}^n$ as defined in \eqref{eq:def:Gi} with true parameter value $\bstheta_0 \in \Theta$ as in \eqref{eq:def:Th} observed on a fixed grid of size $\Delta = 1$. We choose a pure jump Variance Gamma (VG) process as the L\'evy process, which has infinite activity and has been used successfully for modeling stock prices (see \citet{Haug07} and reference therein). The L\'evy measure of the VG process with parameter $C > 0$ has Lebesgue density
\begin{equation}\label{eq:ss1}
\nu_L(\diff x) = \frac{C}{|x|} \exp \{ -(2C)^{1/2}|x| \} \diff x, \quad x \not= 0.
\end{equation}

The Indirect Inference method of \citet{Gourieroux93} based on simulations was originally proposed to estimate models where the binding function is difficult or impossible to compute. However, the binding function $\bstheta \mapsto \bspi_{\bstheta}$ from Proposition~\ref{pr:bind:inj} can be computed explicitly from the formulas {given in Theorem~3.1 of \citet{Haug07}}
and the Yule-Walker equations in \eqref{eq:yw}, leading to the IIE
\begin{equation}\label{eq:ii:no:sim}
\hat{\bm{\theta}}_{n,\iistar} := \argminA_{\bm{\theta} \in \Theta} \| \hatbspi_n - \bspi_{\bstheta} \|_{\bm{\Omega}}.
\end{equation} 
We perform a simulation study to evaluate the finite sample performance of the IIE $\hat{\bm{\theta}}_{n,\iistar}$ in \eqref{eq:ii:no:sim} and also to compare it with the method of moments (MM) estimator $\hat{\bm{\theta}}_{n,\mm}$ (Algorithm 1 in \citet{Haug07}) and the optimal prediction based (OPB) estimator $\hat{\bm{\theta}}_{n,\opb}$ (equation (7) in \citet{Bibbona15}). 
As in the simulation studies in \cite{Bibbona15,Haug07}, we take the VG process with true parameter value $\bstheta_0 = (0.04, 0.053, 0.038)$ and $C = 1$ in \eqref{eq:ss1}, which implies $\Psi_{\bstheta_0}(4) = -0.0261 < 0$. Under these conditions, all three estimators $\hat{\bm{\theta}}_{n,\mm}$ (Theorem~3.8 in \cite{Haug07}), $\hat{\bm{\theta}}_{n,\opb}$ (Theorem 3.1 in \cite{Bibbona15}) and IIE $\hat{\bm{\theta}}_{n,\iistar}$ (Theorem~\ref{th:ii:cons}) are consistent.

The MM is based on $r$ empirical autocovariances, OPB based on $r$ predictors, and IIE based on an AR($r$) auxiliary model. Inspection of several empirical autocovariance functions of the squared returns $\bm{G}^2_n$ with $n = 10\,000$ revealed $r = 70$ as a suitable number of lags in most of the cases. Since we have to fix $r$ in a simulation study, we choose $r = 70$ for all three estimators.

We compare the three estimators in a simulation study in Section~\ref{sec:sim_gen}. Then we show in Section~\ref{sec:fin_bias} how the IIE based on simulations defined in \eqref{eq:DefIIE} can reduce the finite sample bias of $\hat{\bm{\theta}}_{n,\iistar}$ considerably. Finally, to understand how the condition $\Psi_{\bstheta_0}(4) < 0$ affects the estimation, we investigate the finite sample bias of both IIEs with two different true parameter values $\bstheta^{(1)}_0$ and $\bstheta^{(2)}_0$ satisfying $\Psi_{\bstheta^{(2)}_0}(4) < \Psi_{\bstheta_0}(4) < \Psi_{\bstheta^{(1)}_0}(4) < 0$, where $\Psi_{\bstheta^{(1)}_0}(4)$ is near zero.

\subsection{Simulation results}\label{sec:sim_gen}

The computations were performed using the R software (\citet{Rsoftware}). Simulation of the COGARCH process and computation of $\hat{\bm{\theta}}_{n,\mm}$ and $\hat{\bm{\theta}}_{n,\opb}$ are performed with the COGARCH R package from \citet{Bibbona14} (see also the YUIMA R package in \citet{Iacus17cogarch} for the simulation and estimation of higher order COGARCH models).
We first compute $\hat{\bm{\theta}}_{n,\mm}$ based on the sample $\bm{G}^2_n$. 
The estimators  $\hat{\bm{\theta}}_{n,\opb}$ and $\hat{\bm{\theta}}_{n,\iistar}$ are computed via the optimization routine \textit{optim} in R, which requires an initial parameter value and we take $\hat{\bm{\theta}}_{n,\mm}$. 
To compute $\hat{\bspi}_n$ in \eqref{eq:ii:no:sim} we use the YWE from Definition~\ref{de:ywlse} and take the identity matrix for $\bm{\Omega}$ to compute $\hat{\bm{\theta}}_{n,\iistar}$.
In principle, there is an optimal choice of $\Omega$ (see Remark~3 of \citet{deLuna01} and Prop.~4 of \citet{Gourieroux93}).
It depends on the covariance matrix $\bm{\Sigma}$ of the auxiliary model in \eqref{eq:defSig} (see also Remark 4.24(b) in \cite{Fasen18II}). 
This matrix depends on an infinite series and on covariances between COGARCH returns to the powers 2,4,6 and 8, and has no explicit expression.
According to Remark~3 of \cite{deLuna01} and empirical evidence reported on p.~S97f  of \citet{Gourieroux93} the gain of efficiency when using the optimal weight matrix is negligible, so that we only consider estimators based on the identity matrix for $\Omega$.

\begin{table}[t]
\centering
\begin{tabular}{rrrrrr}
  \hline
 & & Mean & Std & RMSE & RB \\ 
  \hline
& $\hat{\beta}$ & 0.04698 & 0.02032 & 0.02148 & 0.17457 \\ 
$\hat{\bm{\theta}}_{n,\iistar}$ & $\hat{\eta}$ & 0.05038 & 0.01482 & 0.01504 & -0.04939 \\ 
& $\hat{\vp}$ & 0.03243 & 0.00994 & 0.01139 & -0.14663 \\ 
\hline
& $\hat{\beta}$ & 0.05226 & 0.01805 & 0.02182 & 0.30658 \\ 
$\hat{\bm{\theta}}_{n,\mm}$ & $\hat{\eta}$ & 0.05662 & 0.01576 & 0.01616 & 0.06827 \\ 
& $\hat{\vp}$ & 0.03667 & 0.01023 & 0.01031 & -0.03513 \\ 
\hline
& $\hat{\beta}$ & 0.04439 & 0.01609 & 0.01667 & 0.10965 \\  
$\hat{\bm{\theta}}_{n,\opb}$ & $\hat{\eta}$ & 0.05274 & 0.01317 & 0.01317 & -0.00489 \\ 
& $\hat{\vp}$ & 0.03583 & 0.00815 & 0.00843 & -0.05712 \\ 
   \hline
\hline  
& $\hat{\beta}$ & 0.04204 & 0.02032 & 0.02041 & 0.05105 \\ 
$\hat{\bm{\theta}}_{n,\ii}$ & $\hat{\eta}$ &  0.05318 & 0.01623 & 0.01622 & 0.00336 \\ 
& $\hat{\vp}$ & 0.03661 & 0.00955 & 0.00965 & -0.03661 \\
\hline
\hline
\end{tabular}
\caption{Performance assessment based on $1\,000$ independent samples of COGARCH squared returns $\bm{G}^2_n$ for $n = 10\,000$, sampled 
with parameter values $\beta_0 = 0.04$, $\eta_0 = 0.053$ and $\vp_0 = 0.038$: mean, standard deviation (Std), root mean squared error (RMSE) and relative bias (RB). Both IIEs $\hat{\bm{\theta}}_{n,\iistar}$ in \eqref{eq:ii:no:sim} and $\hat{\bm{\theta}}_{n,\ii}$ in \eqref{eq:ss_3} used the identity matrix for $\bm{\Omega}$. The IIE $\hat{\bm{\theta}}_{n,\ii}$ is based on $K = 100$ simulated paths.}\label{tb:mc2}
\end{table}

{We focus on the YWE for the auxiliary model, a comparison including the LSE will be given in \citet{Thiago}.}
The estimator $\hat{\bm{\theta}}_{n,\opb}$ only returns a result when $\Psi_{\hat{\bm{\theta}}_{n,\opb}}(4) < 0$. 
The estimators $\hat{\bm{\theta}}_{n,\mm}$ and $\hat{\bm{\theta}}_{n,\iistar}$ always return a value. The results are given in Table~\ref{tb:mc2}, where we excluded those paths for which the condition $\Psi_{\hat\bstheta_n}(4) <  0$ is not satisfied for at least one of the estimators compared here. The results are based on $1\,000$ independent samples of COGARCH squared returns.

The results in Table~\ref{tb:mc2} for the estimators $\hat{\bm{\theta}}_{n,\mm}$ and $\hat{\bm{\theta}}_{n,\opb}$ are similar to those of Table~2 in \citet{Bibbona15}. The OPB estimator has the smallest RMSE. The MM has the smallest relative bias for the parameter $\vp$, and the OPB the smallest for $\beta$ and $\phi$. 
The estimator $\hat{\bm{\theta}}_{n,\iistar}$ performed similarly to $\hat{\bm{\theta}}_{n,\mm}$, but it has a large bias for the parameters $\beta$ and $\vp$. This is probably due to the fact that $\hat{\bm{\theta}}_{n,\iistar}$ depends on $\hatbspi_n$, which is a biased estimator of $\bspi$ even for AR models with i.i.d. noise as shown in \citet{Shaman88}. 
The auxiliary AR model from Proposition~\ref{pr:cog:aux} has stationary and ergodic residuals, and 
certainly $\hatbspi_n$ has a bias, which propagates to the IIE. 
As a remedy, we use the IIE based on simulations and show that it can reduce the bias of $\hat{\bm{\theta}}_{n,\iistar}$, {it also outperforms 
$\hat{\bm{\theta}}_{n,\mm}$ and $\hat{\bm{\theta}}_{n,\opb}$.}

\subsection{Finite sample bias}\label{sec:fin_bias}

\begin{table}[t]
\centering
\begin{tabular}{rrrrrr}
  \hline
   & & \multicolumn{4}{c}{$n = 5\,000$}\\
 \hline
 & & Mean & Std & RMSE & RB \\ 
  \hline
& $\hat{\beta}$ & 0.04710 & 0.02196 & 0.02307 & 0.17738 \\ 
$\hat{\bm{\theta}}_{n,\iistar}$ & $\hat{\eta}$ & 0.04977 & 0.02036 & 0.02061 & -0.06094 \\ 
& $\hat{\vp}$ & 0.03168 & 0.01317 & 0.01461 & -0.16637 \\ 
\hline  
& $\hat{\beta}$ & 0.04999 & 0.03228 & 0.03377 & 0.24968 \\ 
$\hat{\bm{\theta}}_{n,\ii}$ & $\hat{\eta}$ & 0.05935 & 0.02458 & 0.02538 & 0.11974 \\  
& $\hat{\vp}$ & 0.03990 & 0.01379 & 0.01391 & 0.04989 \\ 
  \hline
   & & \multicolumn{4}{c}{$n = 7\,500$}\\
 \hline
 & & Mean & Std & RMSE & RB \\ 
  \hline
& $\hat{\beta}$ & 0.05093 & 0.02015 & 0.02291 & 0.27323 \\ 
$\hat{\bm{\theta}}_{n,\iistar}$ & $\hat{\eta}$ & 0.05401 & 0.01786 & 0.01788 & 0.01896 \\ 
& $\hat{\vp}$ & 0.03439 & 0.01158 & 0.01212 & -0.09502 \\ 
\hline  
& $\hat{\beta}$ & 0.04181 & 0.02375 & 0.02381 & 0.04537 \\ 
$\hat{\bm{\theta}}_{n,\ii}$ & $\hat{\eta}$ & 0.05322 & 0.01897 & 0.01896 & 0.00408 \\  
& $\hat{\vp}$ & 0.03668 & 0.01093 & 0.01101 & -0.03487 \\ 
\hline
\end{tabular}
\caption{Performance assessment based on $1\,000$ independent samples of COGARCH squared returns $\bm{G}^2_n$ for $n = 5\,000$ and $n = 7\,500$,
sampled with parameter values $\beta_0 = 0.04$, $\eta_0 = 0.053$ and $\vp_0 = 0.038$: mean, standard deviation (Std), root mean squared error (RMSE) and relative bias (RB). Both IIEs $\hat{\bm{\theta}}_{n,\iistar}$ in \eqref{eq:ii:no:sim} and $\hat{\bm{\theta}}_{n,\ii}$ in \eqref{eq:ss_3} used the identity matrix for $\bm{\Omega}$. The IIE $\hat{\bm{\theta}}_{n,\ii}$ is based on $K = 100$ simulated paths.}\label{tb:mcn}
\end{table}

In \citet{Gourieroux00, Gourieroux10} it is shown that Indirect Inference based on simulations can reduce the finite sample bias considerably, in particular, when the bias originates from the estimator of the auxiliary model. 
The idea of the bias reduction is that the IIE 
\begin{equation}\label{eq:ss_3}
\hat{\bm{\theta}}_{n,\ii} = \argminA_{\bstheta \in \Theta} \bigg\|\hatbspi_n - \frac{1}{K} \sum_{k=1}^K \hatbspi_{n,k}(\bstheta) \bigg\|_{\bm{\Omega}}, \quad K \in \N,
\end{equation}
from Definition \ref{de:ii} finds a $\bstheta \in \Theta$ which minimizes the distance between two biased estimators, $\hatbspi_n$ and $\frac{1}{K} \sum_{k=1}^K \hatbspi_{n,k}(\bstheta)$. As they have a similar bias, they have a chance to cancel.
We proceed to investigate the finite sample performance of the estimator $\hat{\bm{\theta}}_{n,\ii}$ in \eqref{eq:ss_3}. 

According to \cite{Gourieroux10}, the number of simulated paths $K$ in \eqref{eq:ss_3} has to be large enough to ensure that $\E\,\hatbspi_{n}(\bstheta)$ is well approximated by $\frac{1}{K} \sum_{k=1}^K \hatbspi_{n,k}(\bstheta)$ for all $\bstheta$ appearing in the optimization algorithm. Furthermore, the asymptotic variance of the IIE decreases with $K$ (see Eq.~\eqref{eq:2.22}). 
To compute $\hat{\bm{\theta}}_{n,\ii}$ we need to evaluate the function
\begin{equation}\label{eq:ss_4}
\bstheta \mapsto \frac{1}{K} \sum_{k=1}^K \hatbspi_{n,k}(\bstheta)
\end{equation}
for all $\bstheta$ giving a representation of the parameter space. To compute \eqref{eq:ss_4} for a fixed $\bstheta$, we simulate $K$ independent samples $\bm{G}_{n,k}(\bm{\theta}):=(G_i^{(k)}(\bm{\theta}))_{i=1}^n$ for $k = 1,\dots,K$. For different $\bstheta$ we use the same pseudo-random numbers to generate the $K$ independent samples, which turns \eqref{eq:ss_4} into a deterministic function of $\bstheta$ and thus suitable for optimization.

\begin{figure}
\centering
  \includegraphics[height=10cm, width = 14cm]{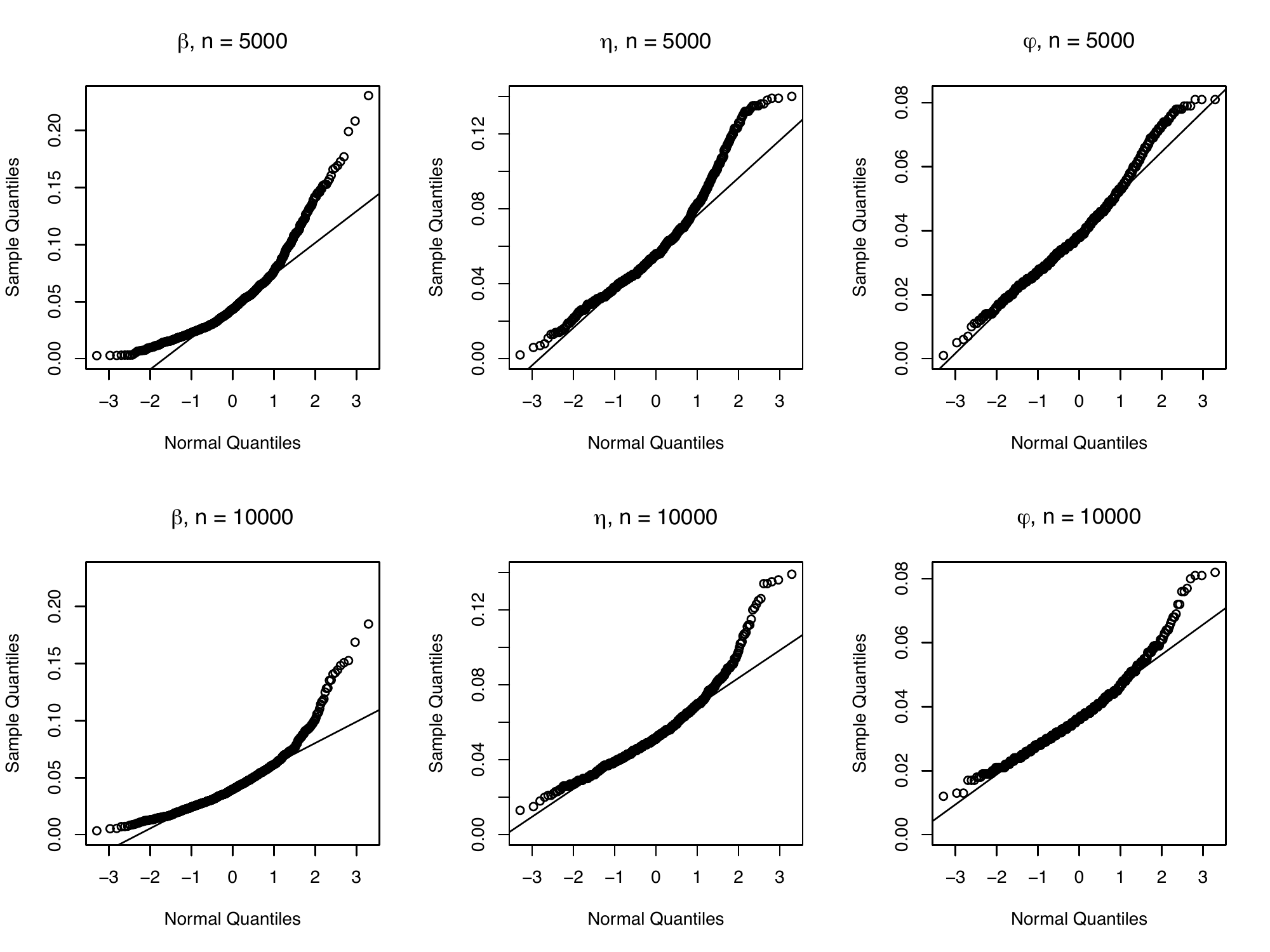}
\caption{QQ plots of the estimators $\hat{\bm{\theta}}_{n,\ii}$ of $\bstheta_0$ as in Table~\ref{tb:mc2} for $n = 5\,000$ (top line) and $n = 10\,000$ (bottom line).}\label{fig:qqnorm}
\end{figure}

In order to save computation time when computing \eqref{eq:ss_4} we use for every simulated path the fact that $\bm{G}_{n,k}(\bm{\theta}) = \sqrt{\beta}\bm{G}_{n,k}((1,\eta,\phi))$ (see Remark~\ref{re:dim2The}) and thus it follows from Definition~\ref{de:ywlse} that
\begin{equation}\label{eq:ss_6}
\hat{\bspi}_{n,k}(\bstheta) = 
\begin{pmatrix}
\hat{\mu}_n(\bstheta) \\
\hat{\bm{a}}_{n,k}(\bstheta) \\
\hat{\gamma}_{n,k}(0;\bstheta)
\end{pmatrix} = \begin{pmatrix}
\beta \hat{\mu}_{n,k}((1,\eta,\phi)) \\
\hat{\bm{a}}_{n,k}((1,\eta,\phi)) \\
\beta^2\hat{\gamma}_{n,k}(0;(1,\eta,\phi))
\end{pmatrix}.
\end{equation}

As it is computationally impossible to perform the optimization \eqref{eq:ss_3} for all $\bstheta \in \Theta$, we have to restrict $\Theta$ in a reasonable way, 
and we restrict $\Theta$ to values in the set 
\beao
\Theta_{\text{rest}} := \{ \bstheta \in \Theta: \Psi_{ \bstheta}(4) < 0, \,\bstheta \in (0,\hat\beta_{\text{max}})\times (0,\hat\eta_{\text{max}}) \times (0,\hat\vp_{\text{max}}) \}
\eeao
where $\hat\beta_{\text{max}}$, $\hat\eta_{\text{max}}$ and $\hat\vp_{\text{max}}$ are upper bounds for the estimated parameters 
 from Table~\ref{tb:mc2} for all $1\,000$ independent samples $\bm{G}_n^2$ and all estimators.

For $K = 100$ and $n = 10\,000$, every evaluation of \eqref{eq:ss_4} takes approximately $13$ minutes on a personal computer. The next goal would be to evaluate \eqref{eq:ss_3} using a gradient based routine. This is out of reach with respect to computation time. As a remedy we adopt the strategy of precomputing \eqref{eq:ss_4} on a fine grid $\Theta_{\text{grid}} \subset \Theta_{\text{rest}}$. The set $\Theta_{\text{grid}}$ was created by generating an equally spaced grid on $\Theta_{\text{rest}}$ with componentwise distance for the parameters $\eta$ and $\vp$ equal to $0.001$ (resulting in about $6.000$ different points). The grid for the component $\beta$ was then created with spacing $0.001$, but without the need to simulate the COGARCH path again by using the relation in \eqref{eq:ss_6}. Afterwards, {with COGARCH returns $\bm{G}^2_n$ generated independently from the samples $\bm{G}_{n,k}(\bm{\theta}), k = 1,\dots,K,$ applied to compute \eqref{eq:ss_4},} we compute $\hatbspi_n$, and the estimator $\hat{\bm{\theta}}_{n,\ii}$ is then simply given by
\begin{equation*}
\argminA_{\bstheta \in \Theta_{\text{grid}}} \bigg\| \hatbspi_n - \frac{1}{K} \sum_{k=1}^K \hatbspi_{n,k}(\bstheta) \bigg\|_{\bm{\Omega}},
\end{equation*}
where we choose the identity matrix for $\bm{\Omega}$. 
The results are presented in the bottom line of Table~\ref{tb:mc2}.
We notice a significant bias reduction for the simulation based estimator $\hat{\bm{\theta}}_{n,\ii}$ compared to $\hat{\bm{\theta}}_{n,\iistar}$. The standard deviation of the estimator for $\eta$ is slightly larger for $\hat{\bm{\theta}}_{n,\ii}$, but this is expected since the simulations increase the asymptotic variance by a factor of $(1 + \frac{1}{K})$ as can be seen from \eqref{eq:2.22}. The relative bias of $\hat{\bm{\theta}}_{n,\ii}$ is also smaller than those of the estimators $\hat{\bm{\theta}}_{n,\mm}$ and $\hat{\bm{\theta}}_{n,\opb}$. 
Since the standard deviations of the components of $\hat{\bm{\theta}}_{n,\ii}$ are larger than for those of $\hat{\bm{\theta}}_{n,\opb}$ and the bias reduction is comparable for the parameters $\eta$ and $\vp$, the RMSE does not seem to improve, even though the bias of $\hat{\bm{\theta}}_{n,\ii}$ is smaller.

We also compare the performance of the IIE with and without simulation for different sample sizes $n$ with $\bstheta_0$ as in Table~\ref{tb:mc2}. The results are given in Table~\ref{tb:mcn}. For $n = 5\,000$ we only observe a bias reduction of $\hat{\bm{\bstheta}}_{n,\ii}$ for $\hat{\eta}$, whereas the bias reduction of $\hat{\bm{\bstheta}}_{n,\ii}$ is noticeable for all three components already for $n = 7\,500$ and of course for $n = 10\,000$; cf. Table~\ref{tb:mc2}. 

We also can see in Figure~\ref{fig:qqnorm} that for $n = 5\,000$ and $n = 10\,000$ the asymptotic normality of $\hat{\bm{\theta}}_{n,\ii}$ has not yet been reached, although some improvement for growing sample sizes is visible in the QQ plots of $\hat{\beta}$ and $\hat{\eta}$, however not for $\hat{\vp}$.

To clarify if the bias reduction of $\hat{\bm{\theta}}_{n,\ii}$ depends on the choice of the true parameter values we perform a simulation study with two different values: $\bstheta_0^{(1)} = (0.04,0.051,0.040)$ and
$\bstheta_0^{(2)} = (0.04,0.055,0.036)$. Both values are in
the stationarity region with 
$$\Psi_{\bstheta_0^{(1)}}(4) = -0.0060,\quad 
\Psi_{\bstheta_0}(4) = -0.0261, 
\quad \Psi_{\bstheta_0^{(2)}}(4) = -0.0460.$$ 
The results are presented in Table~\ref{tb:mc2Th}. 
As for $\bstheta_0$ in Table~\ref{tb:mc2}, they also show significant bias reduction for both values for the estimator $\hat{\bm{\bstheta}}_{n,\ii}$ based on simulations, when compared to $\hat{\bm{\bstheta}}_{n,\iistar}$. However, the bias for $\hat{\beta}^{(1)}$ is much higher than for $\hat{\beta}$ and $\hat{\beta}^{(2)}$ reflecting the fact that $\Psi_{\bstheta_0^{(1)}}(4)$ is very close to zero. The estimators $\hat{\eta}^{(1)}$ and $\hat{\vp}^{(1)}$ seem to be robust with respect to this fact.
Additionally, the relative biases for $\hat{\beta}^{(2)}$ and $\hat{\vp}^{(2)}$ are even smaller than those for $\hat{\beta}$ and $\hat{\vp}$
and $\hat{\beta}^{(1)}$ and $\hat{\vp}^{(1)}$.

\begin{table}[ht]
\centering
\begin{tabular}{rrrrrr}
  \hline
   & & \multicolumn{4}{c}{$\bstheta_0^{(1)} = (0.04,0.051,0.040)$}\\
 \hline
 & & Mean & Std & RMSE & RB \\ 
  \hline
& $\hat{\beta}$ & 0.05452 & 0.02341 & 0.02754 & 0.36298 \\ 
$\hat{\bm{\theta}}_{n,\iistar}$ & $\hat{\eta}$ & 0.05027 & 0.01294 & 0.01296 & -0.01433 \\ 
& $\hat{\vp}$ & 0.03478 & 0.00857 & 0.01003 & -0.13046 \\ 
\hline  
& $\hat{\beta}$ & 0.04586 & 0.02133 & 0.02211 & 0.14658 \\ 
$\hat{\bm{\theta}}_{n,\ii}$ & $\hat{\eta}$ & 0.05142 & 0.01421 & 0.01421 & 0.00827 \\ 
& $\hat{\vp}$ & 0.03788 & 0.00872 & 0.00897 & -0.05300 \\ 
  \hline
   & & \multicolumn{4}{c}{$\bstheta_0^{(2)} = (0.04,0.055,0.036)$}\\
 \hline
 & & Mean & Std & RMSE & RB \\ 
  \hline
& $\hat{\beta}$ & 0.04315 & 0.01858 & 0.01883 & 0.07886 \\ 
$\hat{\bm{\theta}}_{n,\iistar}$ & $\hat{\eta}$ & 0.05177 & 0.01603 & 0.01635 & -0.05867 \\ 
& $\hat{\vp}$ & 0.03109 & 0.01057 & 0.01165 & -0.13643 \\ 
\hline  
& $\hat{\beta}$ & 0.04084 & 0.01829 & 0.01830 & 0.02090 \\ 
$\hat{\bm{\theta}}_{n,\ii}$ & $\hat{\eta}$ & 0.05571 & 0.01666 & 0.01667 & 0.01295 \\ 
& $\hat{\vp}$ & 0.03570 & 0.00948 & 0.00948 & -0.00828 \\ 
\hline
\end{tabular}
\caption{Performance assessment based on $1\,000$ independent samples of COGARCH squared returns $\bm{G}^2_n$ for $n = 10\,000$, sampled 
with parameter values $\bstheta_0^{(1)} = (0.04,0.051,0.040)$ and $\bstheta_0^{(2)} = (0.04,0.055,0.036)$: mean, standard deviation (Std), root mean squared error (RMSE) and relative bias (RB). Both IIEs $\hat{\bm{\theta}}_{n,\iistar}$ in \eqref{eq:ii:no:sim} and $\hat{\bm{\theta}}_{n,\ii}$ in \eqref{eq:ss_3} used the identity matrix for $\bm{\Omega}$. The IIE $\hat{\bm{\theta}}_{n,\ii}$ is based on $K = 100$ simulated paths.
}\label{tb:mc2Th}
\end{table}

%
%
\section*{Acknowledgement}

Thiago do R\^ego Sousa gratefully acknowledges support from CNPq, the National Council for Scientific and Technological Development, Brazil and of the TUM Graduate School. 
All authors take pleasure to thank Carsten Chong, Jean Jacod, Viet Son Pham, and Robert Stelzer for helpful discussions. We thank Vicky Fasen for careful reading and some critical remarks, which lead to an improvement of our paper.


%
%

\appendix

\section{Appendix to Section 4.1}

The first Lemma states important properties about moments of a continuous version of a stochastic process found via Kolmogorov's continuity criterion. The property stated in \eqref{eq:supEkol} is used to apply a uniform SLLN in Theorem \ref{th:ii:cons}. Lemma~\ref{le:sigt:diff} is used to compute $\gradtheta \sigma^2_0(\bstheta)$ and $\gradtheta^2 \sigma^2_0(\bstheta)$, needed to find a continuous version of the map $\bstheta \mapsto \int_{(i-1)\Delta}^{i\Delta} \sigma_s(\bstheta) \diff L_s$ in Theorem~\ref{th:cog:cont}, and of $\bstheta \mapsto \int_{(i-1)\Delta}^{i\Delta} \gradtheta \sigma_s(\bstheta) \diff L_s$ in Lemma~\ref{le:cog:diff}.

\begin{lemma}\label{le:kol_mo}
Let  $(X(\bstheta),\bstheta\in\Theta)$ be a stochastic process with $\Theta \subset \R_{+}^d$ compact for $d\in\N$.
Assume that there exist positive constants $p,c,\epsilon$ such that for all $\bstheta_1,\bstheta_2 \in \Theta$: 
$$
\E  |X^{(c)}(\bm{\theta}_1) - X(\bm{\theta}_2)|^p \leq c  \|\bm{\theta}_1 - \bm{\theta}_2 \|^{d+\epsilon}.
$$
Then there exists a continuous version $(\xcont(\bstheta),\bstheta\in\Theta)$ of $(X(\bstheta),\bstheta\in\Theta)$ such that
\begin{equation}\label{eq:supEkol}
\E \sup_{\theta \in \Theta} |\xcont(\bstheta)|^q < \infty.
\end{equation}
\end{lemma}
\begin{proof}
Since $\Theta$ is compact we can use the Heine-Borel theorem to find a finite collection of open sets $(\Theta_j)_{j=1}^N$ such that $\Theta \subset \cup_{j=1}^N \Theta_j$ and $\|\bm{\theta}_1-\bm{\theta}_2\| \leq \delta^{\ast}$ for every $\bstheta_1, \bstheta_2 \in \Theta_j$. Choosing an arbitrary $\bstheta_j \in \Theta_j \cap \Theta$ for $j=1,\dots,N$ and using $|a-b|^q \leq 2^{q-1}|a^q-b^q|$ gives for $q<p$,
\begin{equation*}
\begin{split}
\E  \sup_{\bstheta \in \Theta} |\xcont(\bstheta)|^q & \leq \sum_{j=1}^N \E \sup_{\bstheta \in \Theta_j} |\xcont(\bstheta)|^q \\
& \leq \sum_{j=1}^N 2 ^{q-1} \E \sup_{\bstheta \in \Theta_j} \big\{ |\xcont(\bstheta) - \xcont(\bstheta_j)|^q + |\xcont(\bstheta_j)|^q \big\} \\
& \leq 
2 ^{q-1} \sum_{j=1}^N  (1 + \E |\xcont(\bstheta_j)|^q) < \infty,
\end{split}
\end{equation*}
since  $\E |\xcont(\bstheta)|^p < \infty$ for all $\bstheta \in \Theta$.
\end{proof}


The next Lemma gives necessary conditions for the existence of a continuous version of a stochastic process, and of fractional moments of order $p \geq 1$ of a random variable that appears when computing inequalities involving moments of $\sigma_0^2(\bstheta)$.

{The following Lemma is well-known from Analysis, and can be found for instance as Exercise~6 in Ch.~15.7 of \cite{Koenigsberger}.}

\begin{lemma}\label{le:pi_f}
Suppose that $g:\R^p \rightarrow \R^q$ is continuous and that
\begin{equation*}
\sup_{\bstheta \in \Theta} \|f_n(\bstheta) - f(\bstheta) \| \stas 0, \quad \nto,
\end{equation*}
where $(f_n(\bstheta))_{n \in \N}$ is a sequence of random vectors in $\R^p$, $f: \Theta \in \R^d \mapsto \R^p$ is a deterministic function and $\Theta$ is compact. Then as $\nto$,
\begin{equation*}
\sup_{\bstheta \in \Theta} \| g(f_n(\bstheta)) - g(f(\bstheta)) \| \stas 0.
\end{equation*}
\end{lemma}

\begin{lemma}\label{le:ws:mom2}
Let $p,b \geq 1$, $a,k \geq 0$, $\bm{\theta} \in \calm$ be fixed and $(K_s(\tildevp))_{s \geq 0}$ as defined in \eqref{eq:def_ds} for fixed $\tildevp > 0$. If $\E |L_1|^{2p(1+\epsilon)} < \infty$ and $\Psi_{\bm{\theta}}(p(1+\epsilon)) < 0$ for some $\epsilon > 0$, then
\begin{equation*}\label{eq:}
\E  \bigg( \int_0^{\infty} (s^a + s^k K^b_s(\tildevp))    e^{-Y_s(\bm{\theta})} \diff s \bigg)^p < \infty.
\end{equation*}
\end{lemma}
\begin{proof}
The proof is similar to the proof of Proposition~4.1 in \citet{LM}. For every $j \in \N_0$ define $Q_j(\bstheta):=  \int_j^{j+1} (s^a + s^k K^b_s(\tildevp))    e^{-Y_s(\bm{\theta})} \diff s$. Then
\beam
\E Q^p_j(\bstheta) & = & \E \Big( \int_j^{j+1} (s^a + s^k K^b_s(\tildevp))    e^{-Y_s(\bm{\theta})} \diff s \Big) ^p \notag\\
& \leq & \E \Big( \sup_{j \leq s \leq j+1} (s^a + s^k K^b_s(\tildevp))    e^{-Y_s(\bm{\theta})}  \Big)^p \notag\\
& \leq & \E\Big( \big( (j+1)^a + (j+1)^k K^b_{j+1}(\tildevp) )^p \sup_{j \leq s \leq j+1}  e^{-pY_s(\bm{\theta})} \Big) \label{eq:Qj:expans}\\
& \leq & \Big( \E \big((j+1)^a + (j+1)^k K^b_{j+1}(\tildevp) \big)^{p(1+\epsilon)/\epsilon} \Big)^{\epsilon/(1+\epsilon)} \Big( \E  \sup_{j \leq s \leq j+1}  e^{-p(1+\epsilon)Y_s(\bm{\theta})}  \Big)^{1/(1+\epsilon)} \notag
\eeam
by the H\"older inequality. Since by Lemma~\ref{le:yt:ineq} $(K_s(\tildevp))_{s \geq 0}$ is a L\'{e}vy process with moments of all orders, repeated differentiation of the characteristic function of $K_{j+1}(\tildevp)$ gives a constant $c > 0$ such that
\begin{equation}\label{eq:E:lhsQj}
\E \big((j+1)^a + (j+1)^k K^b_{j+1}(\tildevp) \big)^{p(1+\epsilon)/\epsilon}  \leq c (j+1)^{mp(1+\epsilon)/\epsilon},
\end{equation}
where $m = a+k+b$.
Since the process $(e^{Y_s(\bstheta) - s\Psi_{\bm{\theta}}(1)})_{s \geq 0}$ is a martingale, we can use Doob's martingale inequality, the Laplace transform in \eqref{eq:mgf:yt} and the fact that $\Psi_{\bm{\theta}}(1) < 0$ to get 
\begin{equation}\label{eq:E:rhsQj}
\begin{split}
\E \sup_{j \leq s \leq j+1}  e^{-p(1+\epsilon)Y_s(\bm{\theta})} & \leq e^{-(j+1)p(1+\epsilon) \Psi_{\bm{\theta}}(1)}  \E  \sup_{j \leq s \leq j+1} e^{-p(1+\epsilon)Y_s(\bm{\theta}) +sp(1+\epsilon) \Psi_{\bm{\theta}}(1)}  \\ 
& \leq e^{-(j+1)p(1+\epsilon) \Psi_{\bm{\theta}}(1)}  \E   e^{-p(1+\epsilon)Y_{j+1}(\bm{\theta}) +p(1+\epsilon)(j+1) \Psi_{\bm{\theta}}(1)}   \\
& =   \E   e^{-p(1+\epsilon)Y_{j+1}(\bm{\theta})}   \\
& =  e^{(j+1)\Psi_{\bm{\theta}}(p(1+\epsilon))}.
\end{split}
\end{equation}
Equation \eqref{eq:Qj:expans} together with \eqref{eq:E:lhsQj} and \eqref{eq:E:rhsQj} gives
\begin{equation}\label{eq:ws:ineq}
\E Q^p_j(\bstheta) \leq c^\ast   (j+1)^{mp} e^{(j+1)\Psi_{\bm{\theta}}(p(1+\epsilon))/(1+\epsilon)}<\infty,
\end{equation}
where $c^\ast = c^{\epsilon/(1+\epsilon)}$.
Let $\alpha:= \lfloor p \rfloor$ be the integer part of $p$ and suppose that $p > \alpha$. 
Then 
\begin{equation}\label{eq:int:0n:ws}
\begin{split}
\Big( \int_0^n (s^a + s^k K^b_s(\tildevp))    e^{-Y_s(\bm{\theta})} \diff s \Big)^p & = \Big( \sum_{j=0}^{n-1} Q_j(\bstheta) \Big)^p \\ 
& = \sum_{j_1=0}^{n-1} \dots \sum_{j_{\alpha} = 0}^{n-1} Q_{j_1}(\bstheta) \dots Q_{j_{\alpha}}(\bstheta) \Big( \sum_{j_{\alpha+1} = 0}^{n-1} Q_{j_{\alpha+1}}(\bstheta) \Big)^{p-\alpha} \\
& \leq  \sum_{j_1=0}^{n-1} \dots \sum_{j_{\alpha} = 0}^{n-1} \sum_{j_{\alpha+1} = 0}^{n-1} Q_{j_1}(\bstheta)  \dots Q_{j_{\alpha}}(\bstheta) Q_{j_{\alpha+1}}^{p - \alpha}(\bstheta).
\end{split}
\end{equation}
If $p$ is an integer the last sum in \eqref{eq:int:0n:ws} disappears. By \eqref{eq:ws:ineq}, for each $j = 1,\dots,\alpha + 1, Q_{j} \in \call^{p}$ so we can apply the H\"{o}lder inequality with $\frac{1}{p} + \dots + \frac{1}{p} + \frac{p-\alpha}{p} = 1$ to the right-hand side of \eqref{eq:int:0n:ws}. This together with \eqref{eq:ws:ineq} gives
\beam
& &  \E \Big( \int_0^n (s^a + s^k K^b_s(\tildevp))    e^{-Y_s(\bm{\theta})} \diff s \Big)^p \notag\\
& \leq &  \sum_{j_1=0}^{n-1} \dots \sum_{j_{\alpha} = 0}^{n-1} \sum_{j_{\alpha+1} = 0}^{n-1} \big( \E (Q_{j_1}^{p}(\bstheta) \big)^{\frac{1}{p}} \dots \big( \E Q_{j_{\alpha}}^{p}(\bstheta) \big)^{\frac{1}{p}} \big( \E (Q_{j_{\alpha+1}}^p(\bstheta) \big)^{\frac{p-\alpha}{p}} \label{eq:lewsm_1}\\
& \leq & c^\ast \Big( \sum_{j=0}^{n-1}   (j + 1)^{m} e^{(j+1)/(p(1+\epsilon)) \Psi_{\bm{\theta}}(p(1+\epsilon))} \Big)^{\alpha} \Big( \sum_{j = 0}^{n-1}  (j + 1)^{m(p-\alpha)} e^{(j+1)(p-\alpha)/(p(1+\epsilon)) \Psi_{\bm{\theta}}(p(1+\epsilon))} \Big).\notag
\eeam
Since $\Psi_{\bm{\theta}}(p(1+\epsilon)) < 0$ both series in \eqref{eq:lewsm_1} converge. The monotone convergence theorem applied to the expectation in the first line of \eqref{eq:lewsm_1} gives the result.
\end{proof}

\begin{lemma}\label{le:sigt:diff}
Let $\bstheta = (\beta,\eta,\vp)$ with $\beta,\eta,\vp > 0$ and consider the process $(Y_s(\bstheta))_{s \geq 0}$ as in \eqref{eq:def:yt}. Let $K_s(\vp)$ be as defined in \eqref{eq:def_ds}. Then:
\begin{itemize}
\item[(a)] For every fixed $s > 0$,
\begin{equation}\label{eq:a.1}
\gradtwo \big( e^{-Y_{s}(\bm{\theta})} \big) = e^{-Y_s(\bstheta)} \begin{pmatrix}  -s  \\  K_s(\vp) \end{pmatrix}.
\end{equation}

\item[(b)] If $\E |L_1|^{2(1+\epsilon)} < \infty$ and $\Psi_{\bstheta}(1+\epsilon) < 0$ for some $\epsilon > 0$, then
\begin{equation}\label{eq:c.1}
\gradtwo \Big( \int_0^{\infty} e^{-Y_s(\bm{\theta})} \diff s \Big) = \int_0^{\infty} e^{-Y_s(\bstheta)} \begin{pmatrix}    -s  \\  K_s(\vp) \end{pmatrix}  \diff s
\end{equation}
and
\begin{equation}\label{eq:c.2}
\gradtwo^2 \Big( \int_0^{\infty} e^{-Y_s(\bm{\theta})} \diff s \Big) = \int_0^{\infty} e^{-Y_s(\bm{\theta})} \begin{pmatrix}
s^2 & -s  K_s(\vp)   \\
 -s  K_s(\vp)  & ( d^2_s(\vp) + d'_s(\vp)  )\\
\end{pmatrix}  \diff s.
\end{equation}
\end{itemize}
\end{lemma}

\begin{proof}
(a) The partial derivatives of $Y_s(\bstheta) = \eta s - \sumzerous \log ( 1 + \vp (\Delta L_u)^2 )$ are given by
$$
\frac{\partial Y_s(\bstheta)}{\partial \eta} = s \quad \text{and} \quad \frac{\partial Y_s(\bstheta)}{\partial \vp} =  -K_s(\vp),
$$
where the derivative with respect to $\vp$ follows by dominated convergence since we have the following bound independent of $\vp$:
\begin{equation}\label{eq:ds_bound}
K_s(\vp) \leq \sumzerous (\Delta L_u)^2 < \infty.
\end{equation}
A simple application of the chain rule gives \eqref{eq:a.1}.\\
(b) It follows from Lemma~\ref{le:yt:ineq}(a) that we can find a collection of points $(\bstheta_j^{\ast})_{j=1}^N$ in $\calm$ such that
\begin{equation}\label{eq:supETh}
\sup_{\bstheta \in \Theta} e^{-Y_s(\bm{\theta})} \leq \sum_{j=1}^N e^{-Y_s(\bstheta_j^{\ast})}, \quad s \geq 0.
\end{equation}
The first derivative of $K_s(\vp)$ follows from dominated converge with the upper bound in \eqref{eq:ds_bound} and is given by
\begin{equation*}\label{eq:ds_der}
K_s'(\vp) = - \sumzerous \frac{ (\Delta L_u)^4}{(1 + \vp(\Delta L_u)^2)^2},\quad s \geq 0.
\end{equation*}
Now, similar calculations as in \eqref{eq:supJump} show that $|K_s'(\vp)| \leq K_s(\vp_\ast)/\vp_\ast$ for $\vp_\ast$ as defined in \eqref{eq:defEVPStar}. This combined with \eqref{eq:supETh} allows us to obtain an upper bound for the sum of the bounds of the absolute values of the integrals at the r.h.s. of \eqref{eq:c.1} and \eqref{eq:c.2} given by
\begin{equation}\label{eq:combo}
 \sum_{j=1}^N \int_0^{\infty} e^{-Y_s(\bstheta_j^{\ast})}\big( s + s^2 + K_s(\vp_\ast)(1 + 2 s + 1/\phi_\ast) + d^2_s(\vp_\ast) \big) \diff s.
\end{equation}
Since $\E |L_1|^{2(1+\epsilon)} < \infty$ and $\Psi_{\bstheta_j^\ast}(1+\epsilon) < 0$ for all $j = 1,\dots,N$ we can apply Lemma~\ref{le:ws:mom2} with $p = 1$ to prove that the integral in \eqref{eq:combo} has finite first moment and is therefore well defined. This allows us to use dominated convergence to differentiate under the integral sign and then use the chain and product rule combined with \eqref{eq:a.1} to obtain \eqref{eq:c.1} and \eqref{eq:c.2}.
\end{proof}

\begin{lemma}\label{le:bEsupSnB}
Let $p \geq 1$ and $k \in \{1,2\}$. If $\E |L_1|^{2kp(1+\epsilon)} < \infty$ for some $\epsilon > 0$ then
\begin{equation*}\label{eq:bEsupa1}
\E \sup_{\bstheta \in \Theta^{(p(1 + \epsilon))} } \| \gradtwo^k \sigma^2_0(\bstheta) \|^p < \infty \quad \text{and} \quad \E \sup_{\bstheta \in \Theta^{(kp(1 + \epsilon))} } \| \gradtwo^k \sigma_0(\bstheta) \|^p < \infty.
\end{equation*}
\end{lemma}
\begin{proof}
For $k \in \{1,2\}$ let $R_{kp}$ denote the integral defined in \eqref{eq:combo} with $(\bstheta_j^\ast)_{j=1}^N \in \calm^{(kp(1 + \epsilon))}$ as in \eqref{eq:def:Thk}. By the same argument preceding \eqref{eq:combo} and from \eqref{eq:def:sig0} we get
\begin{equation}\label{eq:betaR}
 \sup_{\bstheta \in \Theta^{(kp(1 + \epsilon))}}  \Big( \|\gradtwo\sigma^2_0(\bstheta)\| +  \|\gradtwo^2 \sigma^2_0(\bstheta)\| \Big) \leq c \beta^\ast R_{kp}, \quad k = 1,2,
\end{equation}
where $c > 0$ and 
\begin{equation}\label{eq:betaUas}
\beta^\ast = \sup\{\beta > 0: (\beta,\eta,\vp) \in \Theta\} < \infty.
\end{equation}
Since from Lemma~\ref{le:yt:ineq}(b) we know that $\sigma_0(\bstheta) \geq \sigma^{\ast} > 0$, the chain rule implies that
\begin{equation}\label{eq:bouGr1}
\|\gradtwo \sigma_0(\bstheta)\| \leq \frac{1}{\sigma^{\ast}} \|\gradtwo \sigma^2_0(\bstheta)\|.
\end{equation}
Using \eqref{eq:bouGr1} combined with the chain rule for the second order derivative gives
\begin{equation}\label{eq:bouGr2}
\|\gradtwo^2 \sigma_0(\bstheta)\| \leq \frac{1}{4\sigma^{\ast}} \|\gradtwo^2 \sigma^2_0(\bstheta)| + \frac{1}{8(\sigma^{\ast})^3}\|\gradtwo \sigma^2_0(\bstheta)\|^2.
\end{equation}
Using \eqref{eq:betaR} combined with \eqref{eq:bouGr1} gives 
\begin{equation*}
\E \sup_{\bstheta \in \Theta^{(p(1 + \epsilon))}} \| \gradtwo \sigma_0(\bstheta) \|^p \leq  \E \Big(  \frac{1}{\sigma^{\ast}}c\beta^\ast R_{p}  \Big)^p < \infty,
\end{equation*}
by an application of Lemma~\ref{le:ws:mom2}.
Now, \eqref{eq:betaR} combined with \eqref{eq:bouGr2} gives
\begin{equation*}
\E \sup_{\bstheta \in \Theta^{(2p(1 + \epsilon))}} \| \gradtwo^2 \sigma_0(\bstheta) \|^p \leq  \E \Big(  \frac{1}{4\sigma^{\ast}}c\beta^\ast R_{2p}  +  \frac{1}{8(\sigma^{\ast})^3}(c\beta^\ast R_{2p})^2 \Big)^p < \infty,
\end{equation*}
by an application of the Cauchy-Schwartz inequality and Lemma~\ref{le:ws:mom2} with $p$ replaced by $2p$.
\end{proof}

\section{Appendix to Section 4.2}

Lemmas~\ref{le:prop55VF1} and~\ref{le:erg:gra} are used in the proof of Proposition~\ref{prop:b3} to control the convergence of arithmetic means defined in terms of the sequences $(G_i(\bstheta))_{i \in \N}$ and $(\gradtheta G_i(\bstheta))_{i \in \N}$ with $\gradtheta G_i(\bstheta)$ defined in the sense of Remark~\ref{re:diff:G}.

%
%

\begin{lemma}\label{le:erg:gra}
Let $\bstheta = (\beta,\eta,\phi)=:(\theta_1,\theta_2,\theta_3)$ with $\beta,\eta,\phi > 0$ and $\Delta > 0$. 
Suppose that $\E |L_1|^{2} < \infty$ and $\Psi_{\bstheta}(1) < 0$. 
Let $(\si_t(\bstheta))_{t\ge0}$ be the stationary volatility process starting with $\si_0(\bstheta)$ as in \eqref{eq:def:sig0} independent of $L$.
Then for all three components of $\bstheta$ the sequences
\begin{equation*}
\Big(\int_{(i-1)\Delta}^{i\Delta} \sigma_s(\bstheta) \diff L_s ,
\int_{(i-1)\Delta}^{i\Delta}  \frac{\partial}{\partial \theta_j} \sigma_s(\bstheta) \diff L_s \Big)_{i \in \N}
\end{equation*}
 are stationary and ergodic.
\end{lemma}

\begin{proof}
Consider without loss of generality $j = 1$. 
Define the i.i.d. sequence $(S_k)_{k \in \Z}$ with 
$$
S_k = ( \Delta L_u, (k-1)\Delta < u \leq k\Delta).
$$
We consider 
$$ ((\sigma_s(\bstheta),\bstheta\in\Theta), (i-1)\Delta < s \leq i\Delta) =: g(\bstheta,\bstheta\in\Theta, (S_k)_{k =-\infty}^{i})$$
as a measurable function of all relevant jumps $\Delta L_u$. 
Additionally, since limits of differentiable functions are measurable, there exists a measurable map $h$ such that 
$$
 ( \frac{\partial}{\partial \theta_1} \sigma_s(\bstheta), (i-1)\Delta < s \leq i\Delta) = h( (S_k)_{k = -\infty}^{i}, (\bstheta + (c,0,0))_{c \in \Q}).
$$
By observing that a stochastic integral is defined as a measurable map depending on the integrand and integrator processes, we can write
$$
\int_{(i-1)\Delta}^{i\Delta} \sigma_s(\bstheta) \diff L_s = g( (S_k)_{k= -\infty}^{i}, \bstheta) \,\, \text{and} \,\, \int_{(i-1)\Delta}^{i\Delta}  \frac{\partial}{\partial \theta_1} \sigma_s(\bstheta) \diff L_s =  h( (S_k)_{k = -\infty}^{i}, (\bstheta + (c,0,0))_{c \in \Q} ).
$$
Using Proposition 5 in \citet{Straumann06} (see also Theorem~2.1 in \citet{Krengel85}) we can conclude the stationarity and ergodicity of the process $(G_i(\bstheta),\nabla G_i(\bstheta))_{i \in \N}$ based on the stationarity and ergodicity of the sequence $(S_i)_{i \in \Z}$ and the measurability of  $g$ and $h$.
\end{proof}

\begin{lemma}\label{le:prop55VF1}
If $\E |L_1|^{4p(1+\epsilon)} < \infty$ for some $p > 5/2$ and $\epsilon > 0$ then for every $l \in \{1,2,3\}$ and $h \in \N_0$ we have
\begin{equation}\label{eq:supDeco}
\sup_{\bstheta \in \ThetaHCtwo} \Big|\partiall \Big(\frac{1}{n} \sum_{i=1}^{n-h}  G^2_i(\bstheta)G^2_{i+h}(\bstheta) \Big) - \partiall \big(\E G^2_1(\bstheta)G^2_{1+h}(\bstheta)\big) \Big| \stp 0, \quad \nto.
\end{equation}
\end{lemma}
\begin{proof}
The proof follows closely the strategy in the proof of Proposition~5.5 in \citet{Fasen18II}, which divides the proof into three steps: Pointwise convergence, local H\"{o}lder continuity, and stochastic equicontinuity. Let $l \in \{1,2,3\}$ and $h \in \N_0$ be fixed. Write $ \hat{\mu}_n(h;\bstheta) = \frac{1}{n} \sum_{i=1}^{n-h} G^2_i(\bstheta)G^2_{i+h}(\bstheta)$. Then, a simple application of the chain and product rule gives
\begin{equation}\label{eq:dYthen}
\partiall \hat{\mu}_n(h;\bstheta) = \frac{1}{n} \sum_{i=1}^{n-h} \Big[ 2G_i(\bstheta) \Big( \partiall G_i(\bstheta)   \Big) G^2_{i+h}(\bstheta) + 2 G^2_i(\bstheta) G_{i+h}(\bstheta) \Big( \partiall G_{j+h}(\bstheta)\Big)   \Big].
\end{equation}
\textbf{Step 1. Pointwise convergence}. Let $\bstheta \in \ThetaHCtwo$ be fixed. It follows from Lemma~\ref{le:erg:gra} that the sequence $(G_i(\bstheta), \partiall G_i(\bstheta))_{i \in \N}$ is stationary and ergodic. Additionally, it follows from the lemma's assumptions combined with Theorem~\ref{th:cog:cont}, Lemma~\ref{le:cog:diff} and the H\"{older} inequality with $\frac{1}{5} + \frac{2}{5} + \frac{2}{5} = 1$ that
\begin{equation}\label{eq:Edevmpa}
\begin{split}
     \quad\quad & \E G_1(\bstheta) \Big( \partiall G_1(\bstheta)   \Big) G^2_{1+h}(\bstheta)   \\
     & \leq (\E G^5_1(\bstheta))^{1/5} \bigg( \E \Big(  \partiall G_1(\bstheta)\Big)^{5/2} \bigg)^{2/5}  (\E G^5_{1+h}(\bstheta))^{2/5} < \infty.
\end{split}
\end{equation}
The same calculations in \eqref{eq:Edevmpa} can be applied to show that the expectation of the second term in the summation \eqref{eq:dYthen} is also finite. This allows us to apply Birkhoff convergence theorem to conclude that
\begin{equation*}
\partiall \hat{\mu}_n(h;\bstheta) \stp \E G^2_1(\bstheta)G^2_{1+h}(\bstheta), \quad \nto.
\end{equation*}
\\
\textbf{Step 2. $\partiall \hat{\mu}_n(h;\bstheta)$ is locally H\"{o}lder-continuous on $\ThetaHCtwo$}. For $i \in \N$ let $U_i$ and $V_i$ be as defined in \eqref{eq:DefUi} and
\eqref{eq:DefVi}, respectively. By stationarity of $(G_i(\bstheta), \bstheta \in \Theta)_{i \in \N}$ and $(\partiall G_i(\bstheta), \bstheta \in \Theta)_{i \in \N}$, $U_i \eqd U_1$, $V_i \eqd V_1$ and for every $\bstheta_1,\bstheta_2 \in \ThetaHCtwo$ with $\|\bstheta_1-\bstheta_2\| < 1$ it follows from Theorem~\ref{th:cog:cont} and Lemma~\ref{le:cog:diff} that there exists $\gamma \in (0,1)$ such that for all $i \in \N$:
\begin{equation*}\label{eq:Gk123}
|G_i(\bstheta_1) - G_i(\bstheta_2)| \leq U_i \|\bstheta_1-\bstheta_2\|^{\gamma}
\end{equation*}
and 
\begin{equation*}\label{eq:Gk12}
\Big|\partiall G_i(\bstheta_1) - \partiall G_i(\bstheta_2)\Big| \leq V_i \|\bstheta_1-\bstheta_2\|^{\gamma}.
\end{equation*}
Using the inequality
\begin{equation*}\label{eq:abc2}
|a_1 b_1 c_1^2 - a_2 b_2 c_2^2| \leq |a_1||b_1||c_1 + c_2||c_1 - c_2| + |a_1||c_2^2||b_1 - b_2| + |b_2 c_2^2||a_1 - a_2|,
\end{equation*}
valid for every $a_1,a_2,b_1,b_2,c_1,c_2 \in \R$ gives for all $i \in \N$,
\begin{equation}\label{eq:SEFd}
\begin{split}
& \Big|G_i(\bstheta_1) \Big( \partiall G_i(\bstheta_1)   \Big) G^2_{i+h}(\bstheta_1) - G_i(\bstheta_2) \Big( \partiall G_i(\bstheta_2)   \Big) G^2_{i+h}(\bstheta_2) \Big|  \\
& \leq 2\bigg(  \sup_{\bstheta \in \ThetaHCtwo} |G_i(\bstheta)| \Big|\partiall G_i(\bstheta)\Big||G_{i+h} (\bstheta)| \bigg) U_{i+h} \|\bstheta_1-\bstheta_2\|^{\gamma} \\
& + \bigg(  \sup_{\bstheta \in \ThetaHCtwo} |G_i(\bstheta)||G^2_{i+h}(\bstheta)| \bigg) V_i \|\bstheta_1-\bstheta_2\|^{\gamma} \\
& + \bigg(  \sup_{\bstheta \in \ThetaHCtwo} \Big|\partiall G_i(\bstheta)\Big| |G^2_{i+h}(\bstheta)| \bigg) U_i              \|\bstheta_1-\bstheta_2\|^{\gamma} \\
& =: I_{i,h} \|\bstheta_1-\bstheta_2\|^{\gamma}.
\end{split}
\end{equation}
Another application of the H\"{o}lder inequality combined with \eqref{eq:Esu2} and an analogous result for $\gradtheta G_i(\bstheta)$ gives $\E I_{1,h} < \infty$. Similar calculations as in \eqref{eq:SEFd} can be used to show that for all $i \in \N$
\begin{equation}\label{eq:SEFd2}
\Big|G^2_i(\bstheta_1) G_{i+h}(\bstheta_1) \Big( \partiall G_{i+h}(\bstheta_1)   \Big)  - G^2_i(\bstheta_2) G_{i+h}(\bstheta_2) \Big( \partiall G_{j+h}(\bstheta_2) \Big) \Big|  \leq  I^\ast_{i,h} \|\bstheta_1-\bstheta_2\|^{\gamma},
\end{equation}
with $\E I^\ast_{1,h} < \infty$.\\
\textbf{Step 3. Stochastic equicontinuity.} Let $\xi, \nu > 0$ and $0 < \delta < \min\{1, \eta\xi/\E (I_{1,h}+I^{\ast}_{1,h})\}$. Then, it follows from \eqref{eq:dYthen}, \eqref{eq:SEFd}, \eqref{eq:SEFd2} and Markov's inequality that
\begin{equation*}
\P\Bigg( \sup_{\substack{0 < \|\bstheta_1-\bstheta_2\| < \delta \\ \bstheta_1,\bstheta_2 \in \ThetaHCtwo}} \Big|\partiall \hat{\mu}_n(h;\bstheta_1) - \partiall \hat{\mu}_n(h;\bstheta_2) \Big| > \eta \Bigg) \leq \E (I_{1,h}+I^\ast_{1,h}) \frac{\delta^\gamma}{\eta} < \xi.
\end{equation*}
This together with the pointwise convergence in Step 1 allow us to conclude the uniform convergence in \eqref{eq:supDeco} by means of Theorem~10.2 in \citet{Pollard90EP}.
\end{proof}


%
%

\begin{small}
\bibliographystyle{abbrvnat}
\bibliography{references}       
\end{small}
       
\end{document}